\documentclass[9pt,twoside,article]{IEEEtran}

\ifCLASSINFOpdf
  \usepackage[pdftex]{graphicx}
\else
\fi
\usepackage{cite}
\usepackage{amsmath}
\usepackage{amssymb}
\usepackage{amsthm}
\usepackage{amsfonts}       
\usepackage{algorithm}
\usepackage[noend]{algpseudocode}
\usepackage{url}            
\usepackage[utf8]{inputenc} 
\usepackage[T1]{fontenc}    
\usepackage{booktabs}       
\usepackage{nicefrac}       
\usepackage{microtype}      
\usepackage[colorlinks=false, pdfborder={0 0 0.5}]{hyperref}
\usepackage{cleveref}		

\newtheorem{thm}{Theorem}\crefname{thm}{Theorem}{Theorems}
\newtheorem{lem}{Lemma}
\newtheorem{rem}{Remark}
\newtheorem{defn}{Definition}
\newtheorem{cor}{Corollary}\crefname{cor}{Corollary}{Corollaries}

\begin{document}
\allowdisplaybreaks

\title{Analytic Solution of a Delay Differential Equation Arising in Cost Functionals for Systems with Distributed Delays}

\author{
	Suat~Gumussoy,~\IEEEmembership{Member,~IEEE,} and~Murad~Abu-Khalaf,~\IEEEmembership{Member,~IEEE}

	\thanks{Suat Gumussoy is with MathWorks, Natick, MA 01760 USA. {\tt\small sgumusso@mathworks.com}.}

	\thanks{Murad~Abu-Khalaf is with MIT Computer Science and Artificial Intelligence Laboratory, Massachusetts Institute of Technology, Cambridge, MA~02139 USA. {\tt\small murad@mit.edu}.}

}


%



\maketitle

\begin{abstract}
The solvability of a delay differential equation arising in the construction of quadratic cost functionals, \textit{i.e.} Lyapunov functionals, for a linear time-delay system with a constant and a distributed delay is investigated. We present a delay-free auxiliary ordinary differential equation system with algebraically coupled split-boundary conditions, that characterizes the solutions of the delay differential equation and is used for solution synthesis. A spectral property of the time-delay system yields a necessary and sufficient condition for existence and uniqueness of solutions to the auxiliary system, equivalently the delay differential equation. The result is a tractable analytic solution framework to the delay differential equation.
\end{abstract}


%
\IEEEpeerreviewmaketitle

\section{Introduction} \label{Section:Intro}

This investigation considers the time-delay system
\begin{equation} \label{Equation:TimeDelaySystem} 
\dot{x}(t)=A_{0} x(t)+A_{1} x(t-h)+\smallint\limits_{-h}^{0}A_{D} (\theta )x(t+\theta )d\theta  ,   
\end{equation}

\noindent with initial function $\phi (\cdot )$ such that at $t=0, \; x(t+\theta) = \phi (\theta), \; \forall \theta \in [-h,0]$. We study the solvability of the associated system of equations \eqref{Equation:DDEc} abbreviated as (DDEc), a delay differential equation (DDE) \eqref{Equation:DynamicDDE} with coupling conditions \eqref{Equation:SymmetricDDE} and \eqref{Equation:AlgebraicDDE}:
\begin{subequations} \label{Equation:DDEc}
	\begin{gather}
		\dot{P}(\tau )=P(\tau )A_{0} +P(\tau -h)A_{1} +\smallint\limits_{-h}^{0}P(\tau +\theta )A_{D} (\theta )d\theta  ,\label{Equation:DynamicDDE}\\[5pt]  
		P(-\tau )=P(\tau )^{\intercal}, \label{Equation:SymmetricDDE} \\[5pt] 
		\begin{split} \label{Equation:AlgebraicDDE}
			 -Q=A_{0}^{\intercal} P(0)+P(0)A_{0} +A_{1}^{\intercal} P(h)+P(-h)A_{1}\\
			 +\smallint\limits_{-h}^{0}{\left[ A_{D}(\theta )^{\intercal}P(-\theta )+P(\theta )A_{D}(\theta )\right]d\theta },
		\end{split}
	\end{gather}
\end{subequations}
\noindent where $h\ge 0$, and $\tau \in [0,h]$. Equation \eqref{Equation:SymmetricDDE} is a symmetry constraint and \eqref{Equation:AlgebraicDDE} introduces algebraic and integral constraints involving the boundary values. The state $x(t)$ is an $n\mathrm{\times}1$ function of time, and $A_{D} (\cdot )$ and $P(\cdot )$ are $n\mathrm{\times}n$ functions. Moreover, $A_{0} $, $A_{1} $, and $Q=Q^{\intercal} $ are $n\mathrm{\times}n$ constant matrices. $A_{D} (\theta )\buildrel\Delta\over= C_{d}^{} e^{A_{d} \theta } B_{d}^{} $ with $e^{A_{d} \theta } $ a matrix exponential. $C_{d} $, $A_{d} $, and $B_{d} $ are $n\mathrm{\times}n_{d}$, $n_{d}\mathrm{\times}n_{d}$, and $n_{d}\mathrm{\times}n$ constant matrices.

The time-delay system \eqref{Equation:TimeDelaySystem} has a lumped or constant delay, and a distributed delay. It can be expressed in integrated form \cite{Hale1993} \cite{KharitonovBook2013} as
\begin{equation} \label{Equation:StateTrajectory}
x(t,\phi)=\left\{
				\begin{array}{l}
					{\Phi _{0}(t)\phi(0) +\smallint\limits_{-h}^{0}\Phi_{0}(t-h-\alpha )A_{1}\phi(\alpha)d\alpha \hspace{18pt} t\ge 0} \\  
					{+\smallint\limits_{-h}^{0}\smallint\limits_{-h}^{\alpha }\Phi_{0}(t-\alpha +\tau)A_{D}(\tau)d\tau\phi(\alpha )d\alpha,} \\
					{\phi (t), \hspace{130pt} 0>t\ge -h} 
				\end{array}
			\right.
\end{equation}

\noindent where $\Phi_{0} $ is the fundamental matrix of \eqref{Equation:TimeDelaySystem}. If the time-delay system \eqref{Equation:TimeDelaySystem} is stable, then integration over its trajectories yields
\begin{equation} \label{Equation:Cost-to-go} 
	V(\phi )=\smallint\limits_{0}^{\infty }x(t,\phi )^{\intercal} Qx(t,\phi )dt ,   
\end{equation}

\noindent which serves as a Lyapunov functional, and is also a cost-to-go from an initial state $\phi $ with respect to $Q$. Motivated by parameter optimization problems, we are interested in computing $V(\phi )$ without running \eqref{Equation:TimeDelaySystem} forward-in-time to generate a system response. Substituting equation \eqref{Equation:StateTrajectory} in \eqref{Equation:Cost-to-go} gives the following expression
\begin{subequations} \label{Equation:Cost-to-goAlternative}
	\begin{align}
		\begin{split} \label{Equation:Cost-to-go_P} 
			V(\phi )=&\phi(0)^{\intercal}P(0)\phi(0) + \phi(0)^{\intercal}\smallint\limits_{-h}^{0}P(-h-\beta)A_{1}\phi(\beta)d\beta +\cdots \\ 
			&+\smallint\limits_{-h}^{0}\phi (\alpha )^{\intercal} \smallint\limits_{-h}^{\alpha }A_{D}^{\intercal}(\tau _{1})\smallint\limits_{-h}^{0}\smallint\limits_{-h}^{\beta }P(\alpha-\beta-\tau_{1} +\tau_{2} )\\
			&\hspace{110pt} A_{D}(\tau_{2})d\tau_{2}  \phi (\beta )d\beta  d\tau _{1}  d\alpha  ,
		\end{split}\\
		P(\tau )&\buildrel\Delta\over= \smallint\limits_{0}^{\infty }\Phi _{0} (t)^{\intercal} Q\Phi _{0} (t+\tau )dt ,   \label{Equation:PDefinition}
	\end{align}
\end{subequations}

\noindent where for brevity we omitted some terms in \eqref{Equation:Cost-to-go_P}. Computing $V(\phi )$ using \eqref{Equation:Cost-to-go_P} requires $P(\tau )$. It was shown in \cite{SCL2006}, \cite{Huang1989} and \cite{KharitonovBook2013} under different assumptions on the initial function $\phi $ that $P(\tau )$ defined by \eqref{Equation:PDefinition} satisfies dynamic \eqref{Equation:DynamicDDE}, symmetric \eqref{Equation:SymmetricDDE}, and algebraic \eqref{Equation:AlgebraicDDE} relations. Our aim is therefore to solve the DDEc \eqref{Equation:DDEc} by first writing it as an auxiliary ordinary differential equation (ODE) system with algebraically coupled split-boundary conditions abbreviated as (ODEc); then turning the ODEc into an initial value problem. We show a necessary and sufficient condition for the existence and uniqueness of solutions to this ODEc system.

Section \ref{Section:Background} provides background and notation. Section \ref{Section:MainResult} presents the ODEc system in addition to a necessary and sufficient condition to the existence and uniqueness of solutions to the ODEc. Section \ref{Section:AnalyticSolution} presents a framework for analytic solutions to the DDEc through the ODEc. Conclusions are given in Section \ref{Section:Conclusion}.

\section{Background} \label{Section:Background}
\subsection{Origins of the Problem} \label{Section:OriginsProblem}

In \cite{SCL2006}, the solvability of \eqref{Equation:DynamicDDE}-\eqref{Equation:AlgebraicDDE} was studied in the context of stability. In doing so, \cite{SCL2006} proposed solving the delay Lyapunov equation \eqref{Equation:DynamicDDE}-\eqref{Equation:AlgebraicDDE} by rewriting it as an ordinary differential equation system with split-boundary conditions, and solving the resulting equations instead. In \cite{arXiv:1802.06831}, we show that the split-boundary conditions given in \cite{SCL2006} are linearly dependent and do not yield a unique initial value problem. Thus the auxiliary system proposed in \cite{SCL2006} fails to characterize the solutions of the delay Lyapunov equation, and a numerical example therein incorrectly computes $P(\tau )$. Nonetheless, \cite{SCL2006} provides great insights into solving such problems and inspires our work herein.

We arrived at this problem due to our interest in optimizing and tuning the parameters of stable time-delay systems, \textit{cf.} \cite{MarshallGorecki1992}, with a distributed delay. Distributed delays appear naturally in optimal control of time-delay systems \cite{Krasovskii1962}. Evaluating closed-loop performance could be done by integrating the cost over the system trajectories \eqref{Equation:Cost-to-go} or alternatively via \eqref{Equation:Cost-to-go_P} provided $P(\tau )$ \eqref{Equation:PDefinition} is known.

For results earlier than \cite{SCL2006}, we refer the reader to a body of work on Lyapunov stability for different classes of time-delay systems resulting in related DDEs and auxiliary ODEs \cite{Krasovskii1956,Hale1993,Repin1965,Datko1972,CastelanInfante1977,InfanteCastelan1978,Datko1980,Infante1982,Huang1989,Louisell1997,KharitonovPlischke2006}.

\subsection{Notation} \label{Section:Notation}
$\mathbb{R}$ denotes the real line and $\mathbb{C}$ the set of complex numbers. Given matrix $Y\in \mathbb{R}^{p\times q} $, $Y_{*i} $ denotes its \textit{i${}^{th}$} column and $y_{pq} $ the element at the \textit{p${}^{th}$} row and \textit{q${}^{th}$} column. $vec(\cdot )$ stacks the columns of $Y\in \mathbb{R}^{p\times q} $ into a single column. Operator $\otimes $ is the Kronecker product \cite{Brewer1978} where for a matrix $Z$

	\[vec(Y)\buildrel\Delta\over= \left[\begin{array}{c} {Y_{*1} } \\ {Y_{*2} } \\ {\vdots } \\ {Y_{*q} } \end{array}\right], 
	Y\otimes Z\buildrel\Delta\over= \left[\begin{array}{cccc} {y_{11} Z} & {y_{12} Z} & {\ldots } & {y_{1q} Z} \\ {y_{21} Z} & {\ddots } & {} & {\vdots } \\ {\vdots } & {} & {\ddots } & {\vdots } \\ {y_{p1} Z} & {\ldots } & {\ldots } & {y_{pq} Z} \end{array}\right].\]

An all zero entries two-dimensional matrix of context-dependent size is denoted by $\mathbf{0}$, and $0_{n\times m} $ for a fixed \textit{n}$\mathrm{\times}$\textit{m} size. Similarly, $I$ denotes an identity matrix of arbitrary size, and $I_{n} $ a size \textit{n} identity matrix.

\section{Main Result} \label{Section:MainResult}

In \Cref{Section:ODESplitBoundaryConditions}, we introduce the ODEc \eqref{Equation:ODEc}, an auxiliary ordinary differential equation system with algebraically coupled split-boundary conditions, and study its solvability. In \Cref{Section:RelationDDE}, we show the relation between the ODEc \eqref{Equation:ODEc} and the DDEc \eqref{Equation:DDEc} systems.

\subsection{ODE with Coupled Split-Boundary Conditions (ODEc)} \label{Section:ODESplitBoundaryConditions}

The ODEc is defined by the dynamics \eqref{Equation:ODE} and the split-boundary conditions \eqref{Equation:AlgebraicODE}-\eqref{Equation:Boundary3456ODE} for the interval $[0,h]$, where $\forall \tau \in \mathbb{R}$,
\begin{subequations} \label{Equation:ODEc}
	\begin{align}
		\begin{split} \label{Equation:ODE}			
			\dot\Omega_{1}(\tau) =& \Omega_{1}(\tau)A_{0} +\Omega_{2}(\tau)A_{1} +\Omega_{3}(\tau)B_{d} +\Omega_{4}(\tau)B_{d} , \\ 
			\dot\Omega_{2}(\tau) =& -A_{1}^{\intercal} \Omega_{1}(\tau)-A_{0}^{\intercal} \Omega_{2} (\tau )-B_{d}^{\intercal} \Omega_{5}(\tau )-B_{d}^{\intercal} \Omega_{6}(\tau), \\ 
			\dot\Omega_{3}(\tau) =& -\Omega_{3}(\tau)A_{d} +\Omega_{1}(\tau)C_{d}^{} , \\ 
			\dot\Omega_{4}(\tau) =& -\Omega_{4}(\tau)A_{d} -\Omega_{2}(\tau)C_{d}^{} e^{-A_{d} h} , \\ 
			\dot\Omega_{5}(\tau) =& A_{d}^{\intercal} \Omega_{5} (\tau)+\left(C_{d}^{} e^{-A_{d} h} \right)^{\intercal}\Omega_{1}(\tau), \\ 
			\dot\Omega_{6}(\tau) =& A_{d}^{\intercal} \Omega_{6}(\tau)-C_{d}^{\intercal} \Omega_{2}(\tau),
		\end{split}\\[5pt]
		\begin{split} \label{Equation:AlgebraicODE}
			-Q=&\Omega_{1}(0)A_{0} + \Omega_{2}(0)A_{1} + \Omega_{4}(0)B_{d}\\
			  &+ A_{1}^{\intercal}\Omega_{1}(h)+A_{0}^{\intercal}\Omega_{2}(h)+B_{d}^{\intercal}\Omega_{5}(h),
		\end{split}\\[5pt]
		\mathbf{0}=&\Omega_{1}(0)-\Omega_{2}(h), \label{Equation:Boundary12ODE} \\[5pt]
		\mathbf{0}=&\Omega_{3}(0),\; \mathbf{0}=\Omega_{4}(h),\; \mathbf{0}=\Omega _{5}(0),\; \mathbf{0}=\Omega_{6}(h). \label{Equation:Boundary3456ODE}
	\end{align}
\end{subequations}

\noindent Note that $\Omega _{1} (\cdot )$, $\Omega _{2} (\cdot )$ are \textit{n}$\mathrm{\times}$\textit{n}; $\Omega _{3} (\cdot )$, $\Omega _{4} (\cdot )$ are \textit{n}$\mathrm{\times}$\textit{n${}_{d}$}; $\Omega _{5} (\cdot )$, $\Omega _{6} (\cdot )$ are \textit{n${}_{d}$}$\mathrm{\times}$\textit{n}. Therefore, \eqref{Equation:ODE} has \textit{n${}_{s}$} states where
\begin{equation} \label{Equation:ns} 
	n_{s} =2n^{2} +4nn_{d}.   
\end{equation} 

Existence and uniqueness of solutions to \eqref{Equation:ODEc} corresponds to whether there exists a \textit{unique} initial value $\Omega_{1}(0)$, $\Omega_{2}(0)$, $\Omega_{4}(0)$ and $\Omega_{6}(0)$ with $\mathbf{0}=\Omega_{3}(0)=\Omega_{5}(0)$ such that the solution of \eqref{Equation:ODE} at time $h$ satisfies \eqref{Equation:AlgebraicODE}-\eqref{Equation:Boundary3456ODE}. This is addressed in \Cref{Theorem:UniqueODESpectrum} stated after the following lemmas.

\begin{lem} \label{Lemma:CollapseSubsystems}
$\{\omega _{1}^{} (\tau ),\allowbreak\omega _{2}^{} (\tau ),\omega _{3}^{} (\tau ),\allowbreak\omega _{4}^{} (\tau ),\allowbreak\omega _{5}^{} (\tau ),\allowbreak\omega _{6}^{} (\tau )\}$ is a solution to the ODEc \eqref{Equation:ODEc} if and only if $\forall \tau \in \mathbb{R}$, $\{\omega _{1}^{} (\tau ),\allowbreak\omega _{2}^{} (\tau )\}$ satisfies
\begin{subequations} \label{Equation:CollapseSubsystems}
	\begin{align}
	\begin{split} \label{Equation:Collapsed34}
		\dot\Omega_{1}(\tau) =& \Omega_{1}(\tau)A_{0} + \Omega_{2}(\tau)A_{1} + \smallint\limits_{-\tau}^{0}\Omega_{1}(\tau + \theta)A_{D}(\theta)d\theta \\
		&+\smallint\limits_{-h}^{-\tau}\Omega_{2}(\tau + \theta +h)A_{D}(\theta)d\theta,
	\end{split}\\
	\begin{split} \label{Equation:Collapsed56}
		\dot\Omega_{2}(\tau) =& -A_{1}^{\intercal}\Omega_{1}(\tau)-A_{0}^{\intercal}\Omega_{2}(\tau)\\
		&-\smallint\limits_{-h}^{-h+\tau}A_{D}(\theta )^{\intercal}\Omega_{1}(\tau -\theta -h)d\theta \\ 
		&-\smallint\limits_{-h+\tau}^{0}A_{D}(\theta )^{\intercal}\Omega_{2}(\tau - \theta)d\theta,  
	\end{split}\\
	-Q=&\dot\Omega_{1}(0)-\dot\Omega_{2}(h),  \label{Equation:BoundaryDerivative12}  \\
	\mathbf{0}=&\Omega_{1}(0) - \Omega_{2}(h), \label{Equation:Boundary12NoDerivative}
	\end{align}
\end{subequations}
and $\{\omega _{3}^{} (\tau ),\allowbreak\omega _{4}^{} (\tau ),\allowbreak\omega _{5}^{} (\tau ),\allowbreak\omega _{6}^{} (\tau )\}$ satisfies
\begin{subequations} \label{Equation:CollapsedSubsystems}
	\begin{align} 
		\omega_{3}(\tau ) =& \smallint\limits_{-\tau}^{0}\omega_{1}(\tau + \theta )C_{d}^{} e^{A_{d} \theta} d\theta , \label{Equation:Collapsed3}\\
		\omega_{4}(\tau) =& \smallint\limits_{-h}^{-\tau}\omega_{2}(\tau +\theta +h)C_{d} e^{A_{d} \theta } d\theta, \label{Equation:Collapsed4}\\
		\omega_{5} (\tau ) =& \smallint\limits_{-h}^{-h+\tau }\left(C_{d}^{} e^{A_{d} \theta } \right)^{\intercal} \omega _{1} (\tau -\theta -h)d\theta  , \label{Equation:Collapsed5}\\  
		\omega _{6} (\tau )=& \smallint\limits_{-h+\tau }^{0}\left(C_{d}^{} e^{A_{d} \theta } \right)^{\intercal} \omega _{2} (\tau -\theta )d\theta. \label{Equation:Collapsed6}
	\end{align}
\end{subequations}
\end{lem}
\begin{proof}
To show sufficiency, differentiating the terms in \eqref{Equation:CollapsedSubsystems} with respect to $\tau$, it follows that $\{\omega _{3}^{} (\tau ),\allowbreak\omega _{4}^{} (\tau ),\allowbreak\omega _{5}^{} (\tau ),\allowbreak\omega _{6}^{} (\tau )\}$ satisfies
\begin{align*}
	\dot\omega_{3}(\tau) =& -\omega_{3}(\tau)A_{d} +\omega_{1}(\tau)C_{d}^{} , \\ 
	\dot\omega_{4}(\tau) =& -\omega_{4}(\tau)A_{d} -\omega_{2}(\tau)C_{d}^{} e^{-A_{d} h} , \\ 
	\dot\omega_{5}(\tau) =& A_{d}^{\intercal} \omega_{5} (\tau)+\left(C_{d}^{} e^{-A_{d} h} \right)^{\intercal}\omega_{1}(\tau), \\ 
	\dot\omega_{6}(\tau) =& A_{d}^{\intercal} \omega_{6}(\tau)-C_{d}^{\intercal} \omega_{2}(\tau),
\end{align*}
with boundary conditions $\mathbf{0}=\omega_{3}(0), \mathbf{0}=\omega_{4}(h), \mathbf{0}=\omega _{5}(0), \mathbf{0}=\omega_{6}(h)$. Moreover since $\{\omega _{1}^{} (\tau ),\allowbreak\omega _{2}^{} (\tau )\}$ satisfies \eqref{Equation:CollapseSubsystems}, then $\{\omega _{1}^{} (\tau ),\allowbreak\omega _{2}^{} (\tau ),\omega _{3}^{} (\tau ),\allowbreak\omega _{4}^{} (\tau ),\allowbreak\omega _{5}^{} (\tau ),\allowbreak\omega _{6}^{} (\tau )\}$ satisfies the ODEc \eqref{Equation:ODEc}.

To show necessity, if $\{\omega _{1}^{} (\tau ),\allowbreak\omega _{2}^{} (\tau ),\omega _{3}^{} (\tau ),\allowbreak\omega _{4}^{} (\tau ),\allowbreak\omega _{5}^{} (\tau ),\allowbreak\omega _{6}^{} (\tau )\}$ satisfies the ODEc \eqref{Equation:ODEc}, then from the subsystem $\dot{\Omega}_{3}(\tau) = -\Omega_{3}(\tau)A_{d} + \Omega_{1}(\tau)C_{d} , \mathbf{0}=\Omega_{3}(0)$, it follows that $\omega _{3}^{} (\tau )$ satisfies
\begin{align*}
	\dot{\omega }_{3} (\theta )e{}^{A_{d} \theta } +\omega {}_{3} (\theta )A{}_{d} e{}^{A_{d} \theta }  =& \omega _{1} (\theta )C_{d}^{} e^{A_{d} \theta } , \\ 
	\smallint\limits_{0}^{\tau }\frac{d}{d\theta } \left(\omega _{3} (\theta )e^{A_{d} \theta } \right)d\theta  =& {\smallint\limits_{0}^{\tau }\omega _{1} (\theta )C_{d}^{} e^{A_{d} \theta } d\theta  }, 
\end{align*}
and therefore $\omega_{3}(\tau )$ satisfies \eqref{Equation:Collapsed3}. Moreover, from the subsystem $\dot{\Omega}_{4}(\tau) = -\Omega_{4}(\tau)A_{d} -\Omega_{2}(\tau)C_{d} e^{-A_{d} h} , \mathbf{0}=\Omega_{4}(h)$, it follows that $\omega_{4}(\tau)$ satisfies \eqref{Equation:Collapsed4}. Similarly, from the subsystems for $\Omega_{5}$ and $\Omega_{6}$, it follows that $\omega_{5}(\tau)$ satisfies \eqref{Equation:Collapsed5} and $\omega_{6}(\tau)$ satisfies \eqref{Equation:Collapsed6}. Moreover, we obtain \eqref{Equation:Collapsed34} and \eqref{Equation:Collapsed56} by noting that $\{\omega _{1}^{} (\tau ),\allowbreak\omega _{2}^{} (\tau )\}$ satisfies the first two equations of \eqref{Equation:ODE}, with $\{\omega _{3}^{} (\tau ),\allowbreak\omega _{4}^{} (\tau ),\allowbreak\omega _{5}^{} (\tau ),\allowbreak\omega _{6}^{} (\tau )\}$ satisfying \eqref{Equation:CollapsedSubsystems}. Equation \eqref{Equation:BoundaryDerivative12} follows from evaluating at $\tau =0$ the derivatives of the first two equations of \eqref{Equation:ODE} and using \eqref{Equation:AlgebraicODE}. Therefore, $\{\omega _{1}^{} (\tau ),\allowbreak\omega _{2}^{} (\tau )\}$ satisfies \eqref{Equation:CollapseSubsystems}.
\end{proof}

\begin{lem} \label{Lemma:FlippedSolution}
If $\left\{\omega_{1}^{0}(\tau),\omega_{2}^{0}(\tau),\omega_{3}^{0}(\tau),\omega_{4}^{0}(\tau),\omega_{5}^{0}(\tau),\omega_{6}^{0}(\tau)\right\}$ is a solution to the ODEc \eqref{Equation:ODEc}, then 
the following is also a solution
\begin{multline*}
	\{\omega_{1}^{1}(\tau ),\omega_{2}^{1}(\tau ),\omega_{3}^{1}(\tau),\omega_{4}^{1}(\tau),\omega_{5}^{1}(\tau ),\omega_{6}^{1}(\tau )\} = \{\omega_{2}^{0}(h-\tau)^{\intercal},\\ 
	\omega_{1}^{0}(h-\tau )^{\intercal}, \omega_{6}^{0}(h-\tau )^{\intercal},\Omega_{5}^{0}(h-\tau)^{\intercal},\omega_{4}^{0}(h-\tau )^{\intercal}, \omega_{3}^{0}(h-\tau)^{\intercal} \}.
\end{multline*}
\end{lem}

\begin{proof}
The ODE \eqref{Equation:ODE} is satisfied by $\{\omega_{1}^{1}(\tau ),\allowbreak\omega_{2}^{1}(\tau ),\allowbreak\omega_{3}^{1}(\tau),\allowbreak\omega_{4}^{1}(\tau),\allowbreak\omega_{5}^{1}(\tau ),\allowbreak\omega_{6}^{1}(\tau )\}$:
\begin{align*}
	 \dot{\omega }_{1}^{1}(\tau )  =&-\dot{\omega }_{2}^{0}{{(h-\tau )}^{\intercal}} \\ 
	     =&\omega _{1}^{0}{{(h-\tau )}^{\intercal}}A_{1}^{{}}+\omega _{2}^{0}{{(h-\tau )}^{\intercal}}A_{0}^{{}}+\omega _{5}^{0}{{(h-\tau )}^{\intercal}}B_{d}^{{}} \\
	     &+\omega _{6}^{0}{{(h-\tau )}^{\intercal}}B_{d}^{{}} \\ 
	     =&\omega _{1}^{1}(\tau ){{A}_{0}}+\omega _{2}^{1}(\tau ){{A}_{1}}+\omega _{3}^{1}(\tau ){{B}_{d}}+\omega _{4}^{1}(\tau ){{B}_{d}}, \\ 
	\dot{\omega }_{2}^{1}(\tau ) =& -\dot{\omega }_{1}^{0}{{(h-\tau )}^{\intercal}} \\ 
	      =&-A_{1}^{\intercal}\omega _{1}^{1}(\tau )-A_{0}^{\intercal}\omega _{2}^{1}(\tau )-B_{d}^{\intercal}\omega _{5}^{1}(\tau )-B_{d}^{\intercal}\omega _{6}^{1}(\tau ), \\ 
	 \dot{\omega }_{3}^{1}(\tau ) =&-\dot{\omega }_{6}^{0}{{(h-\tau )}^{\intercal}}=-\omega _{6}^{0}{{(h-\tau )}^{\intercal}}A_{d}^{{}}+\omega _{2}^{0}{{(h-\tau )}^{\intercal}}C_{d}^{{}} \\ 
	      =&-\omega _{3}^{1}(\tau ){{A}_{d}}+\omega _{1}^{1}(\tau )C_{d}^{{}}, \\ 
	\dot{\omega }_{4}^{1}(\tau ) =&-\dot{\omega }_{5}^{0}{{(h-\tau )}^{\intercal}}=-\omega _{4}^{1}(\tau ){{A}_{d}}-\omega _{2}^{1}(\tau )C_{d}^{{}}{{e}^{-{{A}_{d}}h}}, \\ 
	\dot{\omega }_{5}^{1}(\tau ) =&-\dot{\omega }_{4}^{0}{{(h-\tau )}^{\intercal}} \\
	     =&A_{d}^{\intercal}\omega _{4}^{0}{{(h-\tau )}^{\intercal}}+{{({{e}^{-{{A}_{d}}h}})}^{\intercal}}C_{d}^{\intercal}\omega _{2}^{0}{{(h-\tau )}^{\intercal}} \\ 
	      =&A_{d}^{\intercal}\omega _{5}^{1}(\tau )+{{({{e}^{-{{A}_{d}}h}})}^{\intercal}}C_{d}^{\intercal}\omega _{1}^{1}(\tau ), \\ 
	\dot{\omega }_{6}^{1}(\tau ) =&-\dot{\omega }_{3}^{0}{{(h-\tau )}^{\intercal}}=A_{d}^{\intercal}\omega _{6}^{1}(\tau )-C_{d}^{\intercal}\omega _{2}^{1}(\tau ).  
\end{align*}

Moreover, the boundary conditions \eqref{Equation:AlgebraicODE}-\eqref{Equation:Boundary3456ODE} are also satisfied:
\begin{align*}
	-Q=&\left[\begin{array}{l} {A_{0}^{\intercal} \omega _{2}^{0} (h)+A_{1}^{\intercal} \omega _{1}^{0} (h)+B_{d}^{\intercal} \omega _{6}^{0} (h)+B_{d}^{\intercal} \omega _{5}^{0} (h)} \\ {+\omega _{2}^{0} (0)A_{1}^{} +\omega _{1}^{0} (0)A_{0}^{} +\omega _{4}^{0} (0)B_{d}^{} +\omega _{3}^{0} (0)B_{d}^{} } \end{array}\right]^{\intercal}\\
	=&\omega_{1}^{1}(0)A_{0} +\omega_{2}^{1}(0)A_{1} +\omega_{3}^{1}(0)B_{d} +\omega_{4}^{1}(0)B_{d}  \\ 
	&+A_{1}^{\intercal} \omega _{1}^{1} (h)+A_{0}^{\intercal} \omega _{2}^{1} (h)+B_{d}^{\intercal} \omega _{5}^{1} (h)+B_{d}^{\intercal} \omega _{6}^{1} (h), \\[5pt] 
	\mathbf{0}=& \omega_{1}^{1}(0)-\omega_{2}^{1}(h)=\omega_{2}^{0}(h)^{\intercal} -\omega_{1}^{0}(0)^{\intercal} ,\\
	\mathbf{0}=&\omega_{3}^{1}(0)=\omega _{6}^{0}(h)^{\intercal},\; \mathbf{0}=\omega_{5}^{1}(0)=\omega _{4}^{0}(h)^{\intercal} , \\ 
	\mathbf{0}=&\omega_{4}^{1}(h)=\omega _{5}^{0}(0)^{\intercal},\; \mathbf{0}=\omega _{6}^{1}(h)=\omega_{3}^{0}(0)^{\intercal}.
\end{align*}
\end{proof}

\begin{lem}	 \label{Lemma:FlippedUnique}
If $\left\{\omega_{1}(\tau ),\omega_{2} (\tau ),\omega_{3} (\tau ),\omega_{4} (\tau ),\omega_{5} (\tau ),\omega_{6}(\tau )\right\}$ is a unique solution to the ODEc \eqref{Equation:ODEc}, then
\begin{multline} \label{Equation:FlippedUnique} 
	\{\omega_{2}(h-\tau )^{\intercal} ,\omega_{1}(h-\tau )^{\intercal} ,\omega_{6} (h-\tau )^{\intercal} ,\omega_{5}(h-\tau )^{\intercal} ,\omega_{4}(h-\tau )^{\intercal}, \\
	\omega _{3}(h-\tau )^{\intercal} \} = \{\omega_{1}(\tau ),\omega_{2}(\tau ),\omega_{3}(\tau ),\omega_{4}(\tau ),\omega_{5}(\tau ),\omega_{6}(\tau )\}.
\end{multline}

\noindent Moreover, it follows that
\begin{equation} \label{Equation:ZeroSymmetric} 
	\omega _{1}(0)^{\intercal} =\omega_{1}(0).
\end{equation} 
\end{lem}
\begin{proof}
From \Cref{Lemma:FlippedSolution}, the left-hand side of \eqref{Equation:FlippedUnique} is a solution to the ODEc \eqref{Equation:ODEc}. The uniquness of solutions then implies that \eqref{Equation:FlippedUnique} holds. Equation \eqref{Equation:ZeroSymmetric} follows from \eqref{Equation:Boundary12ODE} and \eqref{Equation:FlippedUnique}, namely $\omega _{2}(h)=\omega _{1}(0)$ and $\omega _{2}(h)=\omega_{1}(0)^{\intercal} $.
\end{proof}

\begin{lem} \label{Lemma:Zeroing}
	Let $\{\omega _{1}^{} (\tau ),\allowbreak\omega _{2}^{} (\tau ),\omega _{3}^{} (\tau ),\allowbreak\omega _{4}^{} (\tau ),\allowbreak\omega _{5}^{} (\tau ),\allowbreak\omega _{6}^{} (\tau )\}$ be a solution to the ODEc \eqref{Equation:ODEc} for some $Q$. If $(\forall \tau, \omega_{1}(\tau )=\mathbf{0} \; \lor \; \omega_{2}(\tau )=\mathbf{0})$, then $Q=\mathbf{0}$.
\end{lem}

\begin{proof}
From \eqref{Equation:Boundary3456ODE}, $\omega_{3}(0)=\mathbf{0}$ and $\omega_{5}(0)=\mathbf{0}$. Moreover, if $\forall \tau, \omega _{1} (\tau )=\mathbf{0}$, then $\forall \tau, \omega_{3}(\tau ) = \mathbf{0}, \omega_{5}(\tau ) = \mathbf{0}$. The following dynamics remain:
\begin{align*}
	\dot{\omega }_{2} (\tau ) = & -A_{0}^{\intercal} \omega _{2} (\tau )-B_{d}^{\intercal} \omega _{6} (\tau ), \\ 
	\dot{\omega }_{4} (\tau ) = & -\omega _{4} (\tau )A_{d} -\omega _{2} (\tau )C_{d}^{} e^{-A_{d} h} , \\ 
	\dot{\omega }_{6} (\tau ) = & A_{d}^{\intercal} \omega _{6} (\tau )-C_{d}^{\intercal} \omega _{2} (\tau ), \\ 
	\mathbf{0} = & \omega _{1}(0)=\omega_{2} (h), \; \mathbf{0} = \omega _{4} (h),\; \mathbf{0} = \omega_{6}(h), 
\end{align*}

\noindent which when integrated backward or forward from $\tau = h$ results $\forall \tau$, $\omega _{2} (\tau )=\mathbf{0}, \omega _{4}(\tau ) = \mathbf{0}$, and $\omega_{6}(\tau ) = \mathbf{0}$. From \eqref{Equation:AlgebraicODE}, it follows that $Q=\mathbf{0}$.

Similarly, from \eqref{Equation:Boundary3456ODE}, $\omega_{4}(h)=\mathbf{0}$ and $\omega_{6}(h)=\mathbf{0}$. Moreover, if $\forall \tau, \omega _{2} (\tau )=\mathbf{0}$, then $\forall \tau, \omega_{4}(\tau ) = \mathbf{0}, \omega_{6}(\tau ) = \mathbf{0}$. The following dynamics remain:
\begin{align*}
	\dot{\omega }_{1} (\tau ) =& \omega _{1} (\tau )A_{0} +\omega _{3} (\tau )B_{d} ,\\
	\dot{\omega }_{3} (\tau ) =& -\omega _{3} (\tau )A_{d} +\omega _{1} (\tau )C_{d}^{} ,\\
	\dot{\omega }_{5} (\tau ) =& A_{d}^{\intercal} \omega _{5} (\tau )+(e^{-A_{d} h} )^{\intercal} C_{d}^{\intercal} \omega _{1} (\tau ),\\
	\mathbf{0}=& \omega _{2} (h)=\omega _{1} (0),\; \mathbf{0}=\omega _{3} (0),\; \mathbf{0}=\omega _{5} (0),
\end{align*}

\noindent which when integrated backward or forward from $\tau = 0$ results $\forall \tau$, $\omega _{1} (\tau )=\mathbf{0}, \omega _{3}(\tau ) = \mathbf{0}$, and $\omega_{5}(\tau ) = \mathbf{0}$. From \eqref{Equation:AlgebraicODE}, it follows that $Q=\mathbf{0}$.
\end{proof}

\begin{lem} \label{Lemma:QZero}
If $\{\omega _{1}^{} (\tau ),\allowbreak\omega _{2}^{} (\tau ),\omega _{3}^{} (\tau ),\allowbreak\omega _{4}^{} (\tau ),\allowbreak\omega _{5}^{} (\tau ),\allowbreak\omega _{6}^{} (\tau )\}$ is a solution to the ODEc \eqref{Equation:ODEc} for $Q=\mathbf{0}$, then $\forall \tau, \omega _{1}^{} (\tau )=\omega _{2}^{} (\tau +h)$.
\end{lem}
\begin{proof}
Let $\Delta (\tau )=\omega_{1}(\tau )-\omega_{2}(\tau +h)$, we then wish to show that $\forall \tau \in \mathbb{R}, \Delta (\tau )=\mathbf{0}$. From \eqref{Equation:AlgebraicODE}-\eqref{Equation:Boundary12ODE}, it follows  that
\begin{equation} \label{Equation:DeltaZero}
	\Delta (0)=\mathbf{0}, \dot{\Delta }(0)=\mathbf{0}.
\end{equation}

We next find the dynamics governing $\Delta(\tau)$. For $\omega_{1}(\tau)$, we have
\begin{align*}
	 {{{\ddot{\omega }}}_{1}}(\tau ) =&{{{\dot{\omega }}}_{1}}(\tau ){{A}_{0}}+{{{\dot{\omega }}}_{2}}(\tau ){{A}_{1}}+{{{\dot{\omega }}}_{3}}(\tau ){{B}_{d}}+{{{\dot{\omega }}}_{4}}(\tau ){{B}_{d}} \\ 
	 =&{{{\dot{\omega }}}_{1}}(\tau ){{A}_{0}}-A_{1}^{\intercal}{{\omega }_{1}}(\tau ){{A}_{1}}-A_{0}^{\intercal}{{\omega }_{2}}(\tau ){{A}_{1}} \\ 
	& -B_{d}^{\intercal}{{\omega }_{5}}(\tau ){{A}_{1}}-B_{d}^{\intercal}{{\omega }_{6}}(\tau ){{A}_{1}}+{{{\dot{\omega }}}_{3}}(\tau ){{B}_{d}}+{{{\dot{\omega }}}_{4}}(\tau ){{B}_{d}}.
\end{align*}

\noindent From $A_{0}^{\intercal} \omega _{2} (\tau )A_{1} =A_{0}^{\intercal} \dot{\omega }_{1} (\tau )-A_{0}^{\intercal} \omega _{1} (\tau )A_{0} -A_{0}^{\intercal} \omega _{3} (\tau )B_{d} -A_{0}^{\intercal} \omega _{4} (\tau )B_{d} $, we have
\begin{align*}
	\ddot{\omega}_{1}(\tau)  =&\dot{\omega}_{1}(\tau ){{A}_{0}}-A_{0}^{\intercal}{{{\dot{\omega }}}_{1}}(\tau )+A_{0}^{\intercal}{{\omega }_{1}}(\tau ){{A}_{0}} \\
	&+A_{0}^{\intercal}\left[ {{\omega }_{3}}(\tau )+{{\omega }_{4}}(\tau ) \right]{{B}_{d}}-A_{1}^{\intercal}{{\omega }_{1}}(\tau ){{A}_{1}} \\ 
	& -B_{d}^{\intercal}\left[ \omega _{5}(\tau )+\omega_{6}(\tau ) \right]A_{1}+\left[ \dot{\omega}_{3}(\tau )+\dot{\omega}_{4}(\tau ) \right]B_{d}.  
\end{align*}

\noindent Similarly, we have
\begin{align*}
	{{{\ddot{\omega }}}_{2}}(\tau )=&{{{\dot{\omega }}}_{2}}(\tau ){{A}_{0}}-A_{0}^{\intercal}{{{\dot{\omega }}}_{2}}(\tau )+A_{0}^{\intercal}{{\omega }_{2}}(\tau ){{A}_{0}}\\
	&+B_{d}^{\intercal}\left[ {{\omega }_{5}}(\tau )+{{\omega }_{6}}(\tau ) \right]{{A}_{0}} -A_{1}^{\intercal}{{\omega }_{2}}(\tau ){{A}_{1}}\\
	&-A_{1}^{\intercal}\left[ {{\omega }_{3}}(\tau )+{{\omega }_{4}}(\tau ) \right]{{B}_{d}}-B_{d}^{\intercal}\left[ {{{\dot{\omega }}}_{5}}(\tau )+{{{\dot{\omega }}}_{6}}(\tau ) \right].
\end{align*}
\noindent This implies that
\begin{equation} \label{Equation:DeltaDoubleDot}
	\begin{aligned}
		\ddot{\Delta }(\tau )=&\dot{\Delta }(\tau ){{A}_{0}}-A_{0}^{\intercal}\dot{\Delta }(\tau )+A_{0}^{\intercal}\Delta (\tau ){{A}_{0}}-A_{1}^{\intercal}\Delta (\tau ){{A}_{1}} \\ 
		& +A_{0}^{\intercal}\left[ {{\omega }_{3}}(\tau )+{{\omega }_{4}}(\tau ) \right]{{B}_{d}}-B_{d}^{\intercal}\left[ {{\omega }_{5}}(\tau )+{{\omega }_{6}}(\tau ) \right]{{A}_{1}} \\ 
		& +\left[ {{{\dot{\omega }}}_{3}}(\tau )+{{{\dot{\omega }}}_{4}}(\tau ) \right]{{B}_{d}}-B_{d}^{\intercal}\left[ {{\omega }_{5}}(\tau +h)+{{\omega }_{6}}(\tau +h) \right]{{A}_{0}} \\ 
		& +A_{1}^{\intercal}\left[ {{\omega }_{3}}(\tau +h)+{{\omega }_{4}}(\tau +h) \right]{{B}_{d}}\\
		&+B_{d}^{\intercal}\left[ {{{\dot{\omega }}}_{5}}(\tau +h)+{{{\dot{\omega }}}_{6}}(\tau +h) \right].  
	\end{aligned}
\end{equation}

\noindent From \Cref{Lemma:CollapseSubsystems}, $\omega _{3} (\tau )$ to $\omega _{6} (\tau )$ are function of $\omega _{1} (\tau )$ and $\omega _{2} (\tau )$. We therefore have
\begin{align*}
	{{\omega }_{3}}(\tau )+{{\omega }_{4}}(\tau ) =& \smallint\limits_{-\tau }^{0}{\Delta (\theta +\tau )C_{d}^{{}}{{e}^{{{A}_{d}}\theta }}d\theta } \\
	&+\smallint\limits_{-h}^{0}{{{\omega }_{2}}(\theta +\tau +h)C_{d}^{{}}{{e}^{{{A}_{d}}\theta }}d\theta }, \\ 
	{{\omega }_{5}}(\tau )+{{\omega }_{6}}(\tau ) =& \smallint\limits_{-h}^{\tau -h}{{{\left( C_{d}^{{}}{{e}^{{{A}_{d}}\theta }} \right)}^{\intercal}}\Delta (\tau -\theta -h)d\theta } \\
	&+\smallint\limits_{-h}^{0}{{{\left( C_{d}^{{}}{{e}^{{{A}_{d}}\theta }} \right)}^{\intercal}}{{\omega }_{2}}(\tau -\theta )d\theta },  
\end{align*}

\begin{align*}
	{{\omega }_{3}}(\tau +h)+{{\omega }_{4}}(\tau +h) =& \smallint\limits_{-\tau -h}^{-h}{\Delta (\theta +\tau +h)C_{d}^{{}}{{e}^{{{A}_{d}}\theta }}d\theta } \\ 
	& +\smallint\limits_{-h}^{0}{{{\omega }_{1}}(\theta +\tau +h)C_{d}^{{}}{{e}^{{{A}_{d}}\theta }}d\theta }, \\ 
	{{\omega }_{5}}(\tau +h)+{{\omega }_{6}}(\tau +h) =& \smallint\limits_{0}^{\tau }{{{\left( C_{d}^{{}}{{e}^{{{A}_{d}}\theta }} \right)}^{\intercal}}\Delta (\tau -\theta )d\theta } \\ 
	& +\smallint\limits_{-h}^{0}{{{\left( C_{d}^{{}}{{e}^{{{A}_{d}}\theta }} \right)}^{\intercal}}{{\omega }_{1}}(\tau -\theta )d\theta },  
\end{align*}

\begin{align*}
	{{{\dot{\omega }}}_{3}}(\tau ) &={{\omega }_{1}}(0)C_{d}^{{}}{{e}^{-{{A}_{d}}\tau }}+\smallint\limits_{-\tau }^{0}{{{{\dot{\omega }}}_{1}}(\theta +\tau )C_{d}^{{}}{{e}^{{{A}_{d}}\theta }}d\theta }, \\ 
	{{{\dot{\omega }}}_{4}}(\tau ) &=-{{\omega }_{2}}(h)C_{d}^{{}}{{e}^{-{{A}_{d}}\tau }}+\smallint\limits_{-h}^{-\tau }{{{{\dot{\omega }}}_{2}}(\theta +\tau +h)C_{d}^{{}}{{e}^{{{A}_{d}}\theta }}d\theta },  
\end{align*}

\begin{align*}
	{{{\dot{\omega }}}_{3}}(\tau )+{{{\dot{\omega }}}_{4}}(\tau ) =&\smallint\limits_{-\tau }^{0}{\dot{\Delta }(\theta +\tau )C_{d}^{{}}{{e}^{{{A}_{d}}\theta }}d\theta }\\
	&+\smallint\limits_{-h}^{0}{{{{\dot{\omega }}}_{2}}(\theta +\tau +h)C_{d}^{{}}{{e}^{{{A}_{d}}\theta }}d\theta }, \\ 
	{{{\dot{\omega }}}_{5}}(\tau +h)+{{{\dot{\omega }}}_{6}}(\tau +h) =&\smallint\limits_{0}^{\tau }{{{\left( C_{d}^{{}}{{e}^{{{A}_{d}}\theta }} \right)}^{\intercal}}\dot{\Delta }(\tau -\theta )d\theta } \\ 
	& +\smallint\limits_{-h}^{0}{{{\left( C_{d}^{{}}{{e}^{{{A}_{d}}\theta }} \right)}^{\intercal}}{{{\dot{\omega }}}_{1}}(\tau -\theta )d\theta }.
\end{align*}
Therefore, \eqref{Equation:DeltaDoubleDot} becomes
\begin{equation} \label{Equation:DeltaDoubleDotSimplified}
	\begin{aligned} 
		 \ddot{\Delta }(\tau ) &=\dot{\Delta }(\tau ){{A}_{0}}-A_{0}^{\intercal}\dot{\Delta }(\tau )+A_{0}^{\intercal}\Delta (\tau ){{A}_{0}}-A_{1}^{\intercal}\Delta (\tau ){{A}_{1}} \\
		 &+A_{0}^{\intercal}{{\delta }_{1}}(\tau ){{B}_{d}}-B_{d}^{\intercal}{{\delta }_{2}}(\tau ){{A}_{1}}+{{\delta }_{3}}(\tau ){{B}_{d}}\\
		 &-B_{d}^{\intercal}{{\delta }_{4}}(\tau ){{A}_{0}}+A_{1}^{\intercal}{{\delta }_{5}}(\tau ){{B}_{d}}+B_{d}^{\intercal}{{\delta }_{6}}(\tau ) +\varepsilon (\tau ),
	\end{aligned}
\end{equation}

\noindent where
\begin{equation} \label{Equation:SmallDeltas1-6} 
	\begin{aligned}
		{{\delta }_{1}}(\tau ) &=\smallint\limits_{-\tau }^{0}{\Delta (\theta +\tau )C_{d}^{{}}{{e}^{{{A}_{d}}\theta }}d\theta }, \\ 
		{{\delta }_{2}}(\tau ) &=\smallint\limits_{-h}^{\tau -h}{{{\left( C_{d}^{{}}{{e}^{{{A}_{d}}\theta }} \right)}^{\intercal}}\Delta (\tau -\theta -h)d\theta }, \\ 
		{{\delta }_{3}}(\tau ) &=\smallint\limits_{-\tau }^{0}{\dot{\Delta }(\theta +\tau )C_{d}^{{}}{{e}^{{{A}_{d}}\theta }}d\theta },\\
		{{\delta }_{4}}(\tau ) &=\smallint\limits_{0}^{\tau }{{{\left( C_{d}^{{}}{{e}^{{{A}_{d}}\theta }} \right)}^{\intercal}}\Delta (\tau -\theta )d\theta }, \\ 
		{{\delta }_{5}}(\tau ) &=\smallint\limits_{-\tau -h}^{-h}{\Delta (\theta +\tau +h)C_{d}^{{}}{{e}^{{{A}_{d}}\theta }}d\theta },\text{ } \\ 
		{{\delta }_{6}}(\tau ) &=\smallint\limits_{0}^{\tau }{{{\left( C_{d}^{{}}{{e}^{{{A}_{d}}\theta }} \right)}^{\intercal}}\dot{\Delta }(\tau -\theta )d\theta },
	\end{aligned}
\end{equation}

\noindent and
\begin{multline} \label{Equation:EpsilonTau}
		 \varepsilon (\tau ) =\smallint\limits_{-h}^{0}[{{{\dot{\omega }}}_{2}}(\theta +\tau +h)+A_{0}^{\intercal}{{\omega }_{2}}(\theta +\tau +h) \\
		 + A_{1}^{\intercal}{{\omega }_{1}}(\theta +\tau +h)]{{A}_{D}}(\theta )d\theta \\ 
		 +\smallint\limits_{-h}^{0}{{{A}_{D}}{{(\theta )}^{\intercal}}[ {{{\dot{\omega }}}_{1}}(\tau -\theta )-{{\omega }_{1}}(\tau -\theta ){{A}_{0}}-{{\omega }_{2}}(\tau -\theta ){{A}_{1}} ]d\theta }.  
\end{multline}

\noindent Note that from \eqref{Equation:SmallDeltas1-6} $\delta _{1} (0)=\cdots=\delta _{6} (0)=\mathbf{0}$, and that
\begin{equation} \label{Equation:SmallDeltas1-6Derivatives} 
		\begin{aligned}
			\dot{\delta }_{1} (\tau ) &=-\delta _{1} (\tau )A_{d} -\Delta (\tau )C_{d} ,\\
			\dot{\delta }_{2} (\tau ) &=A_{d}^{\intercal} \delta _{2} (\tau )+\left(C_{d}^{} e^{-A_{d} h} \right)^{\intercal} \Delta (\tau ), \\
			\dot{\delta }_{3} (\tau ) &=-\delta _{3} (\tau )A_{d}^{} +\dot{\Delta }(\tau )C_{d} ,\\
			\dot{\delta }_{4} (\tau ) &=A_{d}^{\intercal} \delta _{4} (\tau )+C_{d}^{\intercal} \Delta (\tau ), \\ 
			\dot{\delta }_{5} (\tau ) &=-\delta _{5} (\tau )A_{d}^{} +\Delta (\tau )C_{d}^{} e^{-A_{d} h} ,\\
			\dot{\delta }_{6} (\tau ) &=A_{d}^{\intercal} \delta _{6} (\tau )+C_{d}^{\intercal} \dot{\Delta }(\tau ). 
		\end{aligned}
\end{equation}

\noindent Substitute for $\dot{\omega }_{1} (\tau -\theta )$ and $\dot{\omega }_{2} (\theta +\tau +h)$ in \eqref{Equation:EpsilonTau} using \eqref{Equation:Collapsed34} and \eqref{Equation:Collapsed56} from \Cref{Lemma:CollapseSubsystems} respectively,  to get $\varepsilon (\tau )=B_{d}^{\intercal} \delta _{7} (\tau )B_{d}^{} $ where
\begin{align*}
	{{\delta }_{7}}(\tau ) =& -\smallint\limits_{-h}^{0}{\smallint\limits_{-h}^{\theta +\tau }{{{\left( C_{d}^{{}}{{e}^{{{A}_{d}}\phi }} \right)}^{\intercal}}{{\omega }_{1}}(\theta +\tau -\phi )d\phi }C_{d}^{{}}{{e}^{{{A}_{d}}\theta }}d\theta } \\ 
	& -\smallint\limits_{-h}^{0}{\smallint\limits_{\theta +\tau }^{0}{{{\left( C_{d}^{{}}{{e}^{{{A}_{d}}\phi }} \right)}^{\intercal}}{{\omega }_{2}}(\theta +\tau -\phi +h)d\phi }C_{d}^{{}}{{e}^{{{A}_{d}}\theta }}d\theta } \\ 
	& +\smallint\limits_{-h}^{0}{{{\left( C_{d}^{{}}{{e}^{{{A}_{d}}\theta }} \right)}^{\intercal}}\smallint\limits_{-\tau +\theta }^{0}{{{\omega }_{1}}(\phi +\tau -\theta )C_{d}^{{}}{{e}^{{{A}_{d}}\phi }}d\phi }d\theta } \\ 
	& +\smallint\limits_{-h}^{0}{{{\left( C_{d}^{{}}{{e}^{{{A}_{d}}\theta }} \right)}^{\intercal}}\smallint\limits_{-h}^{-\tau +\theta }{{{\omega }_{2}}(\phi +\tau -\theta +h)C_{d}^{{}}{{e}^{{{A}_{d}}\phi }}d\phi }d\theta }.  
\end{align*}

\noindent Using $\omega _{2}^{} (\tau +h)=\omega _{1}^{} (\tau )-\Delta (\tau )$, we get
\begin{align*}
	{{\delta }_{7}}(\tau ) =&-\smallint\limits_{-h}^{0}{{{\left( C_{d}^{{}}{{e}^{{{A}_{d}}\theta }} \right)}^{\intercal}}\smallint\limits_{-h}^{-\tau +\theta }{\Delta (\phi +\tau -\theta )C_{d}^{{}}{{e}^{{{A}_{d}}\phi }}d\phi }d\theta } \\ 
	& +\smallint\limits_{-h}^{0}{\smallint\limits_{\theta +\tau }^{0}{{{\left( C_{d}^{{}}{{e}^{{{A}_{d}}\phi }} \right)}^{\intercal}}\Delta (\theta +\tau -\phi )d\phi }C_{d}^{{}}{{e}^{{{A}_{d}}\theta }}d\theta }.  
\end{align*}

\noindent Changing the integration limits, we get $\forall \tau \in \mathbb{R}$:
\begin{align*}
	 {{\delta }_{7}}(\tau ) =&\smallint\limits_{\tau -h}^{0}{\smallint\limits_{-h}^{-\tau +\phi }{{{\left( C_{d}^{{}}{{e}^{{{A}_{d}}\phi }} \right)}^{\intercal}}\Delta (\tau -\phi +\theta )C_{d}^{{}}{{e}^{{{A}_{d}}\theta }}d\theta }d\phi } \\ 
	& +\smallint\limits_{0}^{\tau }{\smallint\limits_{0}^{-\tau +\phi }{{{\left( C_{d}^{{}}{{e}^{{{A}_{d}}\phi }} \right)}^{\intercal}}\Delta (\tau -\phi +\theta )C_{d}^{{}}{{e}^{{{A}_{d}}\theta }}d\theta }d\phi } \\ 
	& -\smallint\limits_{\tau -h}^{0}{\smallint\limits_{-h}^{-\tau +\theta }{{{\left( C_{d}^{{}}{{e}^{{{A}_{d}}\theta }} \right)}^{\intercal}}\Delta (\tau +\phi -\theta )C_{d}^{{}}{{e}^{{{A}_{d}}\phi }}d\phi }d\theta } \\ 
	& -\smallint\limits_{-h}^{\tau -h}{\smallint\limits_{-h}^{-\tau +\theta }{{{\left( C_{d}^{{}}{{e}^{{{A}_{d}}\theta }} \right)}^{\intercal}}\Delta (\tau +\phi -\theta )C_{d}^{{}}{{e}^{{{A}_{d}}\phi }}d\phi }d\theta },
\end{align*}
\noindent which simplifies further to
\begin{align*}
{{\delta }_{7}}(\tau ) =&\smallint\limits_{0}^{\tau }{\smallint\limits_{0}^{-\tau +\phi }{{{\left( C_{d}^{{}}{{e}^{{{A}_{d}}\phi }} \right)}^{\intercal}}\Delta (\tau -\phi +\theta )C_{d}^{{}}{{e}^{{{A}_{d}}\theta }}d\theta }d\phi } \\ 
& -\smallint\limits_{-h}^{\tau -h}{\smallint\limits_{-h}^{-\tau +\theta }{{{\left( C_{d}^{{}}{{e}^{{{A}_{d}}\theta }} \right)}^{\intercal}}\Delta (\tau +\phi -\theta )C_{d}^{{}}{{e}^{{{A}_{d}}\phi }}d\phi }d\theta }.  
\end{align*}

\noindent Note that $\delta _{7} (0)=0$. The dynamics of $\delta _{7} (\tau )$ is given by
\begin{align*}
	{{{\dot{\delta }}}_{7}}(\tau ) &=\left. \smallint\limits_{0}^{\tau }{\frac{d}{d\tau }\smallint\limits_{0}^{-\tau +\phi }{{{\left( C_{d}^{{}}{{e}^{{{A}_{d}}\phi }} \right)}^{\intercal}}\Delta (\tau -\phi +\theta )C_{d}^{{}}{{e}^{{{A}_{d}}\theta }}d\theta }d\phi } \right\} \textcircled{\raisebox{-0.9pt} {1}} \\ 
	& \left. -\smallint\limits_{-h}^{\tau -h}{\frac{d}{d\tau }\smallint\limits_{-h}^{-\tau +\theta }{{{\left( C_{d}^{{}}{{e}^{{{A}_{d}}\theta }} \right)}^{\intercal}}\Delta (\tau +\phi -\theta )C_{d}^{{}}{{e}^{{{A}_{d}}\phi }}d\phi }d\theta }. \right\} \textcircled{\raisebox{-0.9pt} {2}}
\end{align*}

\noindent The expression ${{\dot{\delta }}_{7}}(\tau )$ has two parts $\textcircled{\raisebox{-0.9pt} {1}}$ and $\textcircled{\raisebox{-0.9pt} {2}}$ that we simplify each seperately to get
\begin{align*}
	\textcircled{\raisebox{-0.9pt}{1}}  =&\smallint\limits_{0}^{\tau }{\frac{d}{d\tau }\smallint\limits_{0}^{-\tau +\phi }{{{\left( C_{d}^{{}}{{e}^{{{A}_{d}}\phi }} \right)}^{\intercal}}\Delta (\tau -\phi +\theta )C_{d}^{{}}{{e}^{{{A}_{d}}\theta }}d\theta }d\phi } \\ 
	 =&\smallint\limits_{0}^{\tau }{\smallint\limits_{0}^{-\tau +\phi }{{{\left( C_{d}^{{}}{{e}^{{{A}_{d}}\phi }} \right)}^{\intercal}}\frac{d}{d\tau }\Delta (\tau -\phi +\theta )C_{d}^{{}}{{e}^{{{A}_{d}}\theta }}d\theta }d\phi } \\ 
	 =&\smallint\limits_{0}^{\tau }{\smallint\limits_{0}^{-\tau +\phi }{\frac{d}{d\theta }\left( {{\left( C_{d}^{{}}{{e}^{{{A}_{d}}\phi }} \right)}^{\intercal}}\Delta (\tau -\phi +\theta )C_{d}^{{}}{{e}^{{{A}_{d}}\theta }} \right)d\theta }d\phi } \\ 
	& -\smallint\limits_{0}^{\tau }{\smallint\limits_{0}^{-\tau +\phi }{{{\left( C_{d}^{{}}{{e}^{{{A}_{d}}\phi }} \right)}^{\intercal}}\Delta (\tau -\phi +\theta )C_{d}^{{}}{{e}^{{{A}_{d}}\theta }}d\theta }d\phi }{{A}_{d}} \\ 
	 =&-\smallint\limits_{0}^{\tau }{{{\left( C_{d}^{{}}{{e}^{{{A}_{d}}\phi }} \right)}^{\intercal}}\Delta (\tau -\phi +0){{C}_{d}}d\phi } \\ 
	& -\smallint\limits_{0}^{\tau }{\smallint\limits_{0}^{-\tau +\phi }{{{\left( C_{d}^{{}}{{e}^{{{A}_{d}}\phi }} \right)}^{\intercal}}\Delta (\tau -\phi +\theta )C_{d}^{{}}{{e}^{{{A}_{d}}\theta }}d\theta }d\phi }{{A}_{d}},  
\end{align*}
\noindent and
\begin{align*}
	\textcircled{\raisebox{-0.9pt} {2}}  =&\smallint\limits_{-h}^{\tau -h}{{{\left( C_{d}^{{}}{{e}^{{{A}_{d}}\theta }} \right)}^{\intercal}}\Delta (\tau -h-\theta )C_{d}^{{}}{{e}^{-{{A}_{d}}h}}d\theta } \\ 
	& +\smallint\limits_{-h}^{\tau -h}{\smallint\limits_{-h}^{-\tau +\theta }{{{\left( C_{d}^{{}}{{e}^{{{A}_{d}}\theta }} \right)}^{\intercal}}\Delta (\tau +\phi -\theta )C_{d}^{{}}{{e}^{{{A}_{d}}\phi }}d\phi }d\theta }{{A}_{d}}.
\end{align*}

\noindent The dynamics for $\delta _{7} (\tau )$ is therefore governed by
\begin{equation} \label{Equation:SmallDelta7Derivative} 
\begin{aligned}
	{{{\dot{\delta }}}_{7}}(\tau ) =&-\smallint\limits_{0}^{\tau }{{{\left( C_{d}^{{}}{{e}^{{{A}_{d}}\phi }} \right)}^{\intercal}}\Delta (\tau -\phi )C_{d}^{{}}d\phi } \\ 
	& +\smallint\limits_{-h}^{\tau -h}{{{\left( C_{d}^{{}}{{e}^{{{A}_{d}}\theta }} \right)}^{\intercal}}\Delta (\tau -h-\theta )C_{d}^{{}}{{e}^{-{{A}_{d}}h}}d\theta }-{{\delta }_{7}}(\tau ){{A}_{d}} \\ 
	 =&-{{\delta }_{4}}(\tau )C_{d}^{{}}+{{\left( C_{d}^{{}}{{e}^{-{{A}_{d}}h}} \right)}^{\intercal}}{{\delta }_{4}}(\tau )C_{d}^{{}}{{e}^{-{{A}_{d}}h}}-{{\delta }_{7}}(\tau ){{A}_{d}}.  
\end{aligned}
\end{equation}
 
\noindent \Cref{Equation:DeltaDoubleDotSimplified,Equation:SmallDeltas1-6Derivatives,Equation:SmallDelta7Derivative} together with \eqref{Equation:DeltaZero} and the fact that $\delta _{1} (0)=\cdots=\delta _{7} (0)=\mathbf{0}$ form a set of differential equations with zero initial conditions. This implies that $\forall \tau \in \mathbb{R}, \; \Delta (\tau )=\mathbf{0}$.
\end{proof}
\begin{defn} \label{Definition:Spectrum}
\emph{Spectrum condition}: The spectrum of system \eqref{Equation:TimeDelaySystem} is
\[\Lambda =\left\{\forall \lambda :\det \left(\lambda I-A_{0} -e^{-\lambda h} A_{1} -\smallint\limits_{-h}^{0}e^{\lambda \theta } A_{D} (\theta )d\theta  \right)=0\right\}, \] 
and the spectrum condition states that $\forall \lambda \in \Lambda ,-\lambda \notin \Lambda $.
\end{defn}

The next theorem establishes a necessary and sufficient condition for the existence and uniqueness of solutions for the ODEc \eqref{Equation:ODEc}. Its proof is influenced by the single delay case, \textit{cf.} \cite[Chapter~2]{KharitonovBook2013}.

\begin{thm} \label{Theorem:UniqueODESpectrum}
	A unique solution to the ODEc \eqref{Equation:ODEc} exists for all $Q$ if and only if the spectrum condition in \Cref{Definition:Spectrum} is satisfied.
\end{thm}
\begin{proof}
Since \eqref{Equation:ODE} is a linear dynamical system with $n_{s}$ states \eqref{Equation:ns}, the existence and uniqueness of solutions to \eqref{Equation:ODEc} is equivalent to whether there exists a unique $\Omega_{1}(0)$, $\Omega_{2}(0)$, $\Omega_{3}(0)$, $\Omega_{4}(0)$, $\Omega_{5}(0)$ and $\Omega_{6}(0)$ such that the solution of \eqref{Equation:ODE} satisfies the constraints \eqref{Equation:AlgebraicODE}-\eqref{Equation:Boundary3456ODE}. This results in $n_s$ scalar linear algebraic equations with $n_s$ unknowns such as formulated in \eqref{Equation:vecODEcCouplingConditions}. Therefore, the proof utilizes the fact that the solution to \eqref{Equation:ODEc} is unique if and only if for the case $Q=\mathbf{0}$, the trivial solution is the only solution, and thus having a trivial kernel for the underlying linear system of unknowns.

\emph{Necessity}: We show that if the spectrum condition is not satisfied, then there exists a nontrivial solution for \eqref{Equation:ODEc} when $Q=\mathbf{0}$. Therefore, assume that the spectrum condition is not satisfied, then either (a) $\exists \lambda _{1} ,\lambda _{2} \in \Lambda :\lambda _{2} =-\lambda _{1} \ne 0$ or (b) $\exists \lambda _{0} \in \Lambda :\lambda _{0} =0$. For (a), this implies $\exists v_{1}^{} ,v_{2}^{} \in {\rm C}^{n} :v_{1}^{} \ne 0_{n\times 1} ,v_{2}^{} \ne 0_{n\times 1} $ and
	\[\begin{array}{l} {v_{1}^{\intercal} \left[\lambda _{1} I-A_{0} -e^{-\lambda _{1} h} A_{1} -\smallint\limits_{-h}^{0}e^{\lambda _{1} \theta } A_{D} (\theta )d\theta  \right]=0_{1\times n} ,} \\ 
	{v_{2}^{\intercal} \left[-\lambda _{1} I-A_{0} -e^{\lambda _{1} h} A_{1} -\smallint\limits_{-h}^{0}e^{-\lambda _{1} \theta } A_{D} (\theta )d\theta  \right]=0_{1\times n} .} \end{array}\] 

\noindent Let $\omega_{1}(\tau)=e^{\lambda_{1}\tau}v_{2}v_{1}^{\intercal} $ noting that $v_{1} v_{2}^{\intercal} \ne 0_{n\times n} $. Let $\omega _{2} (\tau +h)=\omega _{1} (\tau )$, then
\begin{equation} \label{Equation:SmallOmega1Dot} 
	\begin{aligned}
		{{{\dot{\omega }}}_{1}}(\tau )=&{{\lambda }_{1}}{{e}^{{{\lambda }_{1}}\tau }}{{v}_{2}}v_{1}^{\intercal} \\
		=&{{e}^{{{\lambda }_{1}}\tau }}{{v}_{2}}v_{1}^{\intercal}\left[ {{A}_{0}}+{{e}^{-{{\lambda }_{1}}h}}{{A}_{1}}+\smallint\limits_{-h}^{0}{{{e}^{{{\lambda }_{1}}\theta }}{{A}_{D}}(\theta )d\theta } \right] \\
		 =&{{\omega }_{1}}(\tau ){{A}_{0}}+{{\omega }_{2}}(\tau ){{A}_{1}} +\smallint\limits_{-h}^{-\tau }{{{\omega }_{2}}(\tau +\theta +h){{A}_{D}}(\theta )d\theta }\\ 
		& +\smallint\limits_{-\tau }^{0}{{{\omega }_{1}}(\tau +\theta ){{A}_{D}}(\theta )d\theta },  
	\end{aligned}
\end{equation}

\begin{equation} \label{Equation:SmallOmega2Dot} 
	\begin{aligned}
		 {{{\dot{\omega }}}_{2}}(\tau )&={{\lambda }_{1}}{{e}^{{{\lambda }_{1}}(\tau -h)}}{{v}_{2}}v_{1}^{\intercal} \\ 
		& ={{\left[ -{{A}_{0}}-{{e}^{{{\lambda }_{1}}h}}{{A}_{1}}-\smallint\limits_{-h}^{0}{{{e}^{-{{\lambda }_{1}}\theta }}{{A}_{D}}(\theta )d\theta } \right]}^{\intercal}}{{e}^{{{\lambda }_{1}}(\tau -h)}}{{v}_{2}}v_{1}^{\intercal} \\ 
		& =-A_{0}^{\intercal}{{\omega }_{2}}(\tau )-A_{1}^{\intercal}{{\omega }_{1}}(\tau ) \\ 
		& -\smallint\limits_{-h}^{\tau -h}{{{A}_{D}}{{(\theta )}^{\intercal}}{{\omega }_{1}}(\tau -\theta -h)d\theta }-\smallint\limits_{\tau -h}^{0}{{{A}_{D}}{{(\theta )}^{\intercal}}{{\omega }_{2}}(\tau -\theta )d\theta }.  
	\end{aligned}
\end{equation} 

\noindent Given \eqref{Equation:SmallOmega1Dot}, \eqref{Equation:SmallOmega2Dot} and that
\begin{equation} \label{Equation:SmallOmega1-2Boundary}
	\omega _{2} (h)=\omega _{1} (0),{\rm \; }\dot{\omega }_{2} (h)=\dot{\omega }_{1} (0),  
\end{equation} 
it follows from \Cref{Lemma:CollapseSubsystems} that there exists a nontrivial solution to \eqref{Equation:ODEc} for $Q=\mathbf{0}$. Similarly for (b), $\lambda _{0} =0$ and this implies $\exists v_{0}^{} \in \mathbb{C}^{n} :v_{0}^{} \ne 0_{n\times 1} $ and
	\[v_{0}^{\intercal} \left[A_{0} +A_{1} +\smallint _{-h}^{0}A_{D} (\theta )d\theta  \right]=0_{1\times n} . \] 
Let $\omega _{1} (\tau )=e^{\lambda _{0} \tau } v_{0} v_{0}^{\intercal} =v_{0} v_{0}^{\intercal} $ and $\omega _{2} (\tau +h)=\omega _{1} (\tau )$, hence
\begin{align*}
	{{{\dot{\omega }}}_{1}}(\tau )=&{{0}_{n\times n}}={{v}_{0}}v_{0}^{\intercal}\left[ {{A}_{0}}+{{A}_{1}}+\smallint\limits_{-h}^{0}{{{A}_{D}}(\theta )d\theta } \right]\\
	=&{{\omega }_{1}}(\tau ){{A}_{0}}+{{\omega }_{2}}(\tau ){{A}_{1}} \\ 
	& +\smallint\limits_{-h}^{-\tau }{{{\omega }_{2}}(\tau +\theta +h){{A}_{D}}(\theta )d\theta }+\smallint\limits_{-\tau }^{0}{{{\omega }_{1}}(\tau +\theta ){{A}_{D}}(\theta )d\theta },  \\
	{{{\dot{\omega }}}_{2}}(\tau )=&{{0}_{n\times n}}=-{{\left[ {{A}_{0}}+{{A}_{1}}+\smallint\limits_{-h}^{0}{{{A}_{D}}(\theta )d\theta } \right]}^{\intercal}}{{v}_{0}}v_{0}^{\intercal}\\
	=&-A_{0}^{\intercal}{{\omega }_{2}}(\tau )-A_{1}^{\intercal}{{\omega }_{1}}(\tau )-\smallint\limits_{-h}^{\tau -h}{{{A}_{D}}{{(\theta )}^{\intercal}}{{\omega }_{1}}(\tau -\theta -h)d\theta } \\ 
	& -\smallint\limits_{\tau -h}^{0}{{{A}_{D}}{{(\theta )}^{\intercal}}{{\omega }_{2}}(\tau -\theta )d\theta }.  
\end{align*}
\noindent and a nontrivial solution to \eqref{Equation:ODEc} for $Q=\mathbf{0}$ follows from \Cref{Lemma:CollapseSubsystems}.

\emph{Sufficiency}: By contraposition, we show that if there exists a nontrivial solution to the ODEc \eqref{Equation:ODEc} when $Q=\mathbf{0}$, then the spectrum condition is not satisfied. Let $\{\omega _{1}^{} (\tau ),\allowbreak\omega _{2}^{} (\tau ),\omega _{3}^{} (\tau ),\allowbreak\omega _{4}^{} (\tau ),\allowbreak\omega _{5}^{} (\tau ),\allowbreak\omega _{6}^{} (\tau )\}$ be a nontrivial solution to the ODEc \eqref{Equation:ODEc} with $Q=\mathbf{0}$. Since the ODE \eqref{Equation:ODE} is a linear time-invariant finite dimensional system, $\omega _{i}(\tau )$ is written in terms of the eigenmotions of ODE \eqref{Equation:ODE}
\begin{equation} \label{Equation:SmallOmega-i-Sum}
	\omega _{i} (\tau )=\sum _{j=1}^{j_{\max } }e^{\lambda _{j} \tau } \varphi_{i,j} (\tau ) ,   
\end{equation} 
where $\lambda _{j} $ are the distinct eigenvalues of the ODE \eqref{Equation:ODE}, $j_{\max } \le n_{s} $, and $\varphi_{i,j} (\cdot )$ are polynomial matrices of degree $k_{\max } (j)$ given by
\begin{equation} \label{Equation:SmallV-i-j-Sum} 
	\varphi_{i,j} (\tau )=\sum _{k=0}^{k_{\max } (j)}W_{i,j,k} \tau ^{k}. 
\end{equation} 
The dimension of $\varphi_{1,j} (\cdot )$ and $\varphi_{2,j} (\cdot )$ is $n\mathrm{\times}n$; $\varphi_{3,j} (\cdot )$ and $\varphi_{4,j} (\cdot )$ are $n\mathrm{\times}n_{d}$; and $\varphi_{5,j} (\cdot )$ and $\varphi_{6,j} (\cdot )$ are $n_{d}\mathrm{\times}n$. Lastly, $W_{i,j,k} $ are constant complex matrices with compatible dimensions to $\varphi_{i,j} (\cdot )$.

\Cref{Lemma:Zeroing} requires the nontrivial solution \eqref{Equation:SmallOmega-i-Sum} to be such that $\omega _{1} (\tau )\ne \mathbf{0}$ and $\omega _{2} (\tau )\ne \mathbf{0}$, otherwise $\forall i \in \{1,\hdots,6\},\omega _{i} (\tau )=\mathbf{0}$. From \Cref{Lemma:QZero}, $\omega _{1}^{{}}(\tau )=\omega _{2}^{{}}(\tau +h)$ which implies that
\begin{multline*}
	\sum\limits_{j=1}^{{{j}_{\max }}}{{{e}^{{{\lambda }_{j}}\tau }}\left[ {{\varphi}_{1,j}}(\tau )-{{e}^{{{\lambda }_{j}}h}}{{\varphi}_{2,j}}(\tau +h) \right]}={{0}_{n\times n}}, \\ 
	\Rightarrow \sum\limits_{j=1}^{{{j}_{\max }}}{{{e}^{{{\lambda }_{j}}\tau }}\sum\limits_{k=0}^{{{k}_{\max }}(j)}{\left[ {{W}_{1,j,k}}{{\tau }^{k}}-{{e}^{{{\lambda }_{j}}h}}{{W}_{2,j,k}}{{(\tau +h)}^{k}} \right]}}={{0}_{n\times n}}.  
\end{multline*}
	
\noindent It therefore follows that $\forall j \in \{1,\hdots,j_{max}\}$
\begin{equation} \label{Equation:SmallV-1-j} 
	\varphi_{1,j} (\tau )=e^{\lambda _{j} h} \varphi_{2,j} (\tau +h),
\end{equation} 
and that
\begin{align*}
	{{W}_{1,j,{{k}_{\max }}(j)}} &={{e}^{{{\lambda }_{j}}h}}{{W}_{2,j,{{k}_{\max }}(j)}}, \\ 
	{{W}_{1,j,({{k}_{\max }}(j)-1)}}&=h{{e}^{{{\lambda }_{j}}h}}{{W}_{2,j,({{k}_{\max }}(j)-1)}}, \\ 
	\vdots &  \\ 
	{{W}_{1,j,0}}&={{e}^{{{\lambda }_{j}}h}}{{W}_{2,j,0}}.
\end{align*}

\noindent Substituting $\omega _{2}(\tau )=\omega _{1}(\tau -h)$ in \eqref{Equation:Collapsed34}, \Cref{Lemma:CollapseSubsystems} yields
	\[\dot{\omega }_{1} (\tau )=\omega _{1} (\tau )A_{0} +\omega _{1} (\tau -h)A_{1} +\smallint\limits_{-h}^{0}\omega _{1} (\theta +\tau )A_{D} (\theta )d\theta ,\] 

\noindent which by substituting \eqref{Equation:SmallV-i-j-Sum} results in
\begin{multline*}
	0_{n\times n} = \sum\limits_{j=1}^{{j}_{\max}}[{{\lambda }_{j}}{{e}^{{{\lambda }_{j}}\tau }}{{\varphi}_{1,j}}(\tau )+{{e}^{{{\lambda }_{j}}\tau }}{{{\dot{v}}}_{1,j}}(\tau )\\
	-{{e}^{{{\lambda }_{j}}\tau }}{{\varphi}_{1,j}}(\tau ){{A}_{0}}-{{e}^{{{\lambda }_{j}}(\tau -h)}}{{\varphi}_{1,j}}(\tau -h){{A}_{1}}\\
	-\smallint\limits_{-h}^{0}{{{e}^{{{\lambda }_{j}}(\theta +\tau )}}{{\varphi}_{1,j}}(\theta +\tau ){{A}_{D}}(\theta )d\theta}].
\end{multline*}

\noindent Collecting the terms associated with $e^{\lambda _{j} \tau } $, it follows that
\begin{multline} \label{Equation:AllDistinctEigenvalues} 
	{{0}_{n\times n}} = {{\lambda }_{j}}{{\varphi}_{1,j}}(\tau )+{{{\dot{\varphi}}}_{1,j}}(\tau )-{{\varphi}_{1,j}}(\tau ){{A}_{0}}\\
	-{{e}^{-{{\lambda }_{j}}h}}{{\varphi}_{1,j}}(\tau -h){{A}_{1}} -\smallint\limits_{-h}^{0}{{{e}^{{{\lambda }_{j}}\theta }}{{\varphi}_{1,j}}(\theta +\tau ){{A}_{D}}(\theta )d\theta }.  
\end{multline} 

\noindent Note that
\begin{equation} \label{Equation:w1NotZero} 
	\omega _{1} (\tau )\ne \mathbf{0}\Rightarrow \exists j_{0} :\varphi_{1,j_{0} } (\tau )\ne \mathbf{0},  
\end{equation} 

\noindent and substituting \eqref{Equation:SmallV-i-j-Sum} in \eqref{Equation:AllDistinctEigenvalues} yields for the $e^{\lambda _{j_{0} } \tau } $ terms
\begin{multline} \label{Equation:lambdaj0}
	{{0}_{n\times n}}=\sum\limits_{k=0}^{{{k}_{\max }}({{j}_{0}})} [ {{\lambda }_{{{j}_{0}}}}{{W}_{1,{{j}_{0}},k}}{{\tau }^{k}}+k{{W}_{1,{{j}_{0}},k}}{{\tau }^{k-1}}-{{W}_{1,{{j}_{0}},k}}{{A}_{0}}{{\tau }^{k}}\\
	-{{e}^{-{{\lambda }_{{{j}_{0}}}}h}}{{W}_{1,{{j}_{0}},k}}{{A}_{1}}{{(\tau -h)}^{k}} \\ 
	 -\smallint\limits_{-h}^{0}{{{e}^{{{\lambda }_{{{j}_{0}}}}\theta }}{{W}_{1,{{j}_{0}},k}}{{A}_{D}}(\theta ){{(\theta +\tau )}^{k}}d\theta }	].
\end{multline}

\noindent Collecting the ${\tau }^{k_{\max } (j_{0} )}$ terms, it must be that
\begin{equation} \label{Equation:kmaxj0} 
	\begin{aligned}
		{{0}_{n\times n}} &={{\lambda }_{{{j}_{0}}}}{{W}_{1,{{j}_{0}},{{k}_{\max }}({{j}_{0}})}}-{{W}_{1,{{j}_{0}},{{k}_{\max }}({{j}_{0}})}}{{A}_{0}} \\ 
		& -{{e}^{-{{\lambda }_{{{j}_{0}}}}h}}{{W}_{1,{{j}_{0}},{{k}_{\max }}({{j}_{0}})}}{{A}_{1}}-\smallint\limits_{-h}^{0}{{{e}^{{{\lambda }_{{{j}_{0}}}}\theta }}{{W}_{1,{{j}_{0}},{{k}_{\max }}({{j}_{0}})}}{{A}_{D}}(\theta )d\theta } \\ 
		& ={{W}_{1,{{j}_{0}},{{k}_{\max }}({{j}_{0}})}}\left[   {{\lambda }_{{{j}_{0}}}}I-{{A}_{0}}-{{e}^{-{{\lambda }_{{{j}_{0}}}}h}}{{A}_{1}} \right. \\
		&\left. -\smallint\limits_{-h}^{0}{{{e}^{{{\lambda }_{{{j}_{0}}}}\theta }}{{A}_{D}}(\theta )d\theta } \right]  .  
	\end{aligned}
\end{equation}
Note that $W_{1,j_{0} ,k_{\max } (j_{0} )} \ne 0_{n\times n} $ because \eqref{Equation:SmallV-i-j-Sum} is of degree $k_{\max } (j_{0} )$. Equation \eqref{Equation:kmaxj0} implies that there exists a nonzero row of $W_{i,j_{0} ,k_{\max } (j_{0} )} $ which is orthogonal to every column of
\begin{equation} \label{Equation:lambdaj0Column} 
	\lambda _{j_{0} } I-A_{0} -e^{-\lambda _{j_{0} } h} A_{1} -\smallint\limits_{-h}^{0}e^{\lambda _{j_{0} } \theta } A_{D} (\theta )d\theta  .  
\end{equation} 
This requires the dimension of the column space of \eqref{Equation:lambdaj0Column} to be less than $n$, and it must be that
	\[\det \left(\lambda _{j_{0} } I-A_{0} -e^{-\lambda _{j_{0} } h} A_{1} -\smallint\limits_{-h}^{0}e^{\lambda _{j_{0} } \theta } A_{D} (\theta )d\theta  \right)=0.\] 
Hence $\lambda _{j_{0} } $ is an eigenvalue of \eqref{Equation:TimeDelaySystem} in addition to being an eigenvalue of the ODE \eqref{Equation:ODE}. Next it is shown that $-\lambda _{j_{0} } $ is also an eigenvalue of \eqref{Equation:TimeDelaySystem}. Substituting $\omega _{2}(\tau +h)=\omega _{1}(\tau )$ in \eqref{Equation:Collapsed56}, \Cref{Lemma:CollapseSubsystems} yields
	\[\dot{\omega }_{2} (\tau )=-A_{1}^{\intercal} \omega _{2}^{} (\tau +h)-A_{0}^{\intercal} \omega _{2} (\tau )-\smallint\limits_{-h}^{0}A_{D} (\theta )^{\intercal} \omega _{2} (\tau -\theta )d\theta  .\] 
Note that from \eqref{Equation:SmallV-1-j} and \eqref{Equation:w1NotZero}, we have for $j_{0}$
\begin{align*}
	{{\varphi}_{1,{{j}_{0}}}}(\tau ) &={{e}^{{{\lambda }_{{{j}_{0}}}}h}}{{\varphi}_{2,{{j}_{0}}}}(\tau +h),\\
	{{W}_{2,{{j}_{0}},{{k}_{\max }}({{j}_{0}})}}&={{e}^{-{{\lambda }_{{{j}_{0}}}}h}}{{W}_{1,{{j}_{0}},{{k}_{\max }}({{j}_{0}})}}\ne {{0}_{n\times n}}.
\end{align*}	

\noindent Similar to \eqref{Equation:AllDistinctEigenvalues}, we get 
\begin{multline*}
	{{0}_{n\times n}}=-{{\lambda }_{j_0}}{{\varphi}_{2,j_0}}(\tau )-{{{\dot{v}}}_{2,j_0}}(\tau )-A_{0}^{\intercal}{{\varphi}_{2,j_0}}(\tau )-{{e}^{{{\lambda }_{j_0}}h}}A_{1}^{\intercal}{{\varphi}_{2,j_0}}(\tau +h)\\
	-\smallint\limits_{-h}^{0}{{{A}_{D}}{{(\theta )}^{\intercal}}{{e}^{-{{\lambda }_{j_0}}\theta }}{{\varphi}_{2,j_0}}(\tau -\theta )d\theta }.
\end{multline*}

\noindent Repeating steps as in \eqref{Equation:lambdaj0} and \eqref{Equation:kmaxj0}, we get
\begin{multline*}
	{{W}_{2,{{j}_{0}},{{k}_{\max }}({{j}_{0}})}}\left( -{{\lambda }_{{{j}_{0}}}}I-A_{0}^{\intercal}-{{e}^{{{\lambda }_{{{j}_{0}}}}h}}A_{1}^{\intercal} \right. \\
	\left. -\smallint\limits_{-h}^{0}{{{A}_{D}}{{(\theta )}^{\intercal}}{{e}^{-{{\lambda }_{{{j}_{0}}}}\theta }}d\theta } \right)={{0}_{n\times n}},
\end{multline*}

\noindent and hence
	\[\det \left(-\lambda _{j_{0} } I-A_{0}^{\intercal} -e^{\lambda_{j_{0} }h} A_{1}^{\intercal} -\smallint\limits_{-h}^{0}A_{D} (\theta )^{\intercal} e^{-\lambda _{j_{0} } \theta } d\theta  \right)=0.\] 
\noindent Hence $-\lambda _{j_{0} } $ is an eigenvalue of \eqref{Equation:TimeDelaySystem}, and the spectrum condition is not satisfied. Note that if $\lambda _{j_{0} } =0=-\lambda _{j_{0} } $, then we get
\begin{align*}
	{{0}_{n\times n}}={{W}_{1,{{j}_{0}},{{k}_{\max }}({{j}_{0}})}}\left( -{{A}_{0}}-{{A}_{1}}-\smallint\limits_{-h}^{0}{{{A}_{D}}(\theta )d\theta } \right),  
\end{align*}
and
	\[\det \left(-A_{0} -A_{1} -\smallint\limits_{-h}^{0}A_{D} (\theta )d\theta  \right)=0,\] 
\noindent hence $0 \in \Lambda$ and the spectrum condition is not satisfied.
\end{proof}

\subsection{Relation to the DDEc:} \label{Section:RelationDDE}
The DDE \eqref{Equation:DynamicDDE} requires $\forall \tau \in [0,h]$ values of $P(\cdot )$ from $[-h,0]$ acting as an initial function. Our objective is to eliminate the requirement for an initial function and rely on initial conditions only. A step in that direction is to use the symmetry property \eqref{Equation:SymmetricDDE} to rewrite the DDE \eqref{Equation:DynamicDDE} such that it reads from the positive time interval $[0,h]$ only. This introduces counter flow terms that we eliminate by introducing auxiliary variables. The problem then becomes an initial value problem for a linear time-invariant system. This is the spirit of the next results.

\begin{lem} \label{Lemma:DDEsolnODEsoln}
If a solution to the DDEc \eqref{Equation:DDEc} exists, then a solution to the ODEc \eqref{Equation:ODEc} exists.
\end{lem}
\begin{proof}
We show that from an existing DDEc \eqref{Equation:DDEc} solution, we construct a solution to the ODEc \eqref{Equation:ODEc} on $[0,h]$ then extend it to $\forall \tau \in \mathbb{R}$. 

First, consider the terms of \eqref{Equation:DynamicDDE} and apply the symmetry property \eqref{Equation:SymmetricDDE} to write the counterflow backward-in-time term $P(\tau -h)=P(h-\tau )^{\intercal}$. Similarly for
\begin{multline*}
	\smallint\limits_{-h}^{0}{P(\tau +\theta ){{A}_{D}}(\theta )d\theta }=\smallint\limits_{-\tau }^{0}{P(\underbrace{\tau +\theta }_{\ge 0}){{A}_{D}}(\theta )d\theta }\\
	+\smallint\limits_{-h}^{-\tau }{P(\underbrace{\tau +\theta }_{\le 0}){{A}_{D}}(\theta )d\theta },
\end{multline*}

\noindent and noting that
	\[\smallint\limits_{-h}^{-\tau }{P(\underbrace{\tau +\theta }_{\le 0}){{A}_{D}}(\theta )d\theta }=\smallint\limits_{-h}^{-\tau }{P{{(\underbrace{-\tau -\theta }_{\ge 0})}^{\intercal}}{{A}_{D}}(\theta )d\theta },\]
one can write
\begin{align*}
	\smallint\limits_{-h}^{0}{P(\tau +\theta ){{A}_{D}}(\theta )d\theta } &=\smallint\limits_{-\tau }^{0}{P(\tau +\theta ){{A}_{D}}(\theta )d\theta } \\ 
	& +\smallint\limits_{-h}^{-\tau }{P{{(h-(\tau +\theta +h))}^{\intercal}}{{A}_{D}}(\theta )d\theta }.
\end{align*}

\noindent Equation \eqref{Equation:DynamicDDE} now includes two counterflow terms becoming
\begin{equation} \label{Equation:TwoCounterFlowTerms} 
	\begin{aligned}
		\dot{P}(\tau )&=P(\tau ){{A}_{0}}+P{{(h-\tau )}^{\intercal}}{{A}_{1}}+\smallint\limits_{-\tau }^{0}{P(\tau +\theta ){{A}_{D}}(\theta )d\theta } \\ 
		& +\smallint\limits_{-h}^{-\tau }{P{{(h-(\tau +\theta +h))}^{\intercal}}{{A}_{D}}(\theta )d\theta }.
	\end{aligned}
\end{equation}

\noindent The dynamics governing both $P(\tau )$ and $P(h-\tau )^{\intercal} $ is
\begin{subequations} \label{Equation:CounterflowDynamics} 
	\begin{align}
		\begin{split} \label{Equation:CounterflowDynamicsFirst}		
			\dot{P}(\tau )=&P(\tau ){{A}_{0}}+P{{(h-\tau )}^{\intercal}}{{A}_{1}}\\
			&+\smallint\limits_{-\tau }^{0}{P(\tau +\theta )C_{d}^{{}}{{e}^{{{A}_{d}}\theta }}d\theta }B_{d}^{{}} \\ 
			& +\smallint\limits_{-h}^{-\tau }{P{{(h-(\tau +\theta +h))}^{\intercal}}C_{d}^{{}}{{e}^{{{A}_{d}}\theta }}d\theta }B_{d}^{{}}, 
		\end{split}\\
		\begin{split} \label{Equation:CounterflowDynamicsSecond}
			\frac{d}{d\tau }P{{(h-\tau )}^{\intercal}}=&-{{A}_{1}}^{\intercal}P(\tau )-{{A}_{0}}^{\intercal}P{{(h-\tau )}^{\intercal}} \\ 
			& -B_{d}^{\intercal}\smallint\limits_{-h}^{-h+\tau }{{{\left( C_{d}^{{}}{{e}^{{{A}_{d}}\theta }} \right)}^{\intercal}}P(\tau -\theta -h)d\theta } \\ 
			& -B_{d}^{\intercal}\smallint\limits_{-h+\tau }^{0}{{{\left( C_{d}^{{}}{{e}^{{{A}_{d}}\theta }} \right)}^{\intercal}}P{{(h-(\tau -\theta ))}^{\intercal}}d\theta }.
		\end{split}
	\end{align}
\end{subequations}

If $p(\tau)$ is a solution to \eqref{Equation:DDEc}, then it also satisfies \eqref{Equation:CounterflowDynamics}. Let $\{\omega _{1}^{} (\tau ),\allowbreak\omega _{2}^{} (\tau ),\omega _{3}^{} (\tau ),\allowbreak\omega _{4}^{} (\tau ),\allowbreak\omega _{5}^{} (\tau ),\allowbreak\omega _{6}^{} (\tau )\}$ be such that
\begin{align*}
	{\omega_1}(\tau ) =& p(\tau ), \\
	{\omega_2}(\tau ) =& p{{(h-\tau )}^{\intercal}}, \\ 
	{\omega_3}(\tau ) =& \smallint\limits_{-\tau }^{0}{p(\tau +\theta )C_{d}^{{}}{{e}^{{{A}_{d}}\theta }}d\theta }, \\
	{\omega_4}(\tau ) =& \smallint\limits_{-h}^{-\tau }{p{{(h-(\tau +\theta +h))}^{\intercal}}C_{d}^{{}}{{e}^{{{A}_{d}}\theta }}d\theta }, \\ 
	{\omega_5}(\tau ) =& \smallint\limits_{-h}^{-h+\tau }{{{\left( C_{d}^{{}}{{e}^{{{A}_{d}}\theta }} \right)}^{\intercal}}p(\tau -\theta -h)d\theta },\\
	{\omega_6}(\tau ) =& \smallint\limits_{-h+\tau }^{0}{{{\left( C_{d}^{{}}{{e}^{{{A}_{d}}\theta }} \right)}^{\intercal}}p{{(h-(\tau -\theta ))}^{\intercal}}d\theta },  
\end{align*}

\noindent we then have the following dynamics
\begin{subequations} \label{Equation:CounterflowDynamicsAuxiliary}
	\begin{align}
		\begin{split} \label{Equation:CounterflowDynamicsAuxiliaryA}	
			{{\dot{\omega}}_{3}}(\tau )&=\frac{d}{d\tau }\smallint\limits_{-\tau }^{0}{p(\tau +\theta )C_{d}^{{}}{{e}^{{{A}_{d}}\theta }}d\theta }\\
			&=-{\omega_3}(\tau ){{A}_{d}}+{\omega_1}(\tau )C_{d}^{{}},
		\end{split}\\
		{{\dot{\omega}}_{4}}(\tau )&=-{\omega_4}(\tau ){{A}_{d}}-{\omega_2}(\tau )C_{d}^{{}}{{e}^{-{{A}_{d}}h}}, \label{CounterflowDynamicsAuxiliaryB} \\
		{{\dot{\omega}}_{5}}(\tau )&=A_{d}^{\intercal}{\omega_5}(\tau )+{{\left( C_{d}^{{}}{{e}^{-{{A}_{d}}h}} \right)}^{\intercal}}{\omega_1}(\tau ), \label{CounterflowDynamicsAuxiliaryC} \\
		{{\dot{\omega}}_{6}}(\tau )&=A_{d}^{\intercal}{\omega_6}(\tau )-C_{d}^{\intercal}{\omega_2}(\tau ).\label{Equation:CounterflowDynamicsAuxiliaryD}
	\end{align}
\end{subequations}

\noindent The dynamics \eqref{Equation:CounterflowDynamics} and \eqref{Equation:CounterflowDynamicsAuxiliary} together give
\begin{equation} \label{Equation:ResultingODEcDynamics} 
	\begin{aligned}
		&{{{\dot{\omega}}}_{1}}(\tau )={\omega_1}(\tau ){{A}_{0}}+{\omega_2}(\tau ){{A}_{1}}+{\omega_3}(\tau ){{B}_{d}}+{\omega_4}(\tau ){{B}_{d}}, \\ 
		&{{{\dot{\omega}}}_{2}}(\tau )=-{{A}_{1}}^{\intercal}{\omega_1}(\tau )-{{A}_{0}}^{\intercal}{\omega_2}(\tau )-B_{d}^{\intercal}{\omega_5}(\tau )-B_{d}^{\intercal}{\omega_6}(\tau ), \\ 
		&{{{\dot{\omega}}}_{3}}(\tau )=-{\omega_3}(\tau ){{A}_{d}}+{\omega_1}(\tau )C_{d}^{{}}, \\ 
		&{{{\dot{\omega}}}_{4}}(\tau )=-{\omega_4}(\tau ){{A}_{d}}-{\omega_2}(\tau )C_{d}^{{}}{{e}^{-{{A}_{d}}h}}, \\ 
		&{{{\dot{\omega}}}_{5}}(\tau )=A_{d}^{\intercal}{\omega_5}(\tau )+{{\left( C_{d}^{{}}{{e}^{-{{A}_{d}}h}} \right)}^{\intercal}}{\omega_1}(\tau ), \\ 
		&{{{\dot{\omega}}}_{6}}(\tau )=A_{d}^{\intercal}{\omega_6}(\tau )-C_{d}^{\intercal}{\omega_2}(\tau ).
	\end{aligned}
\end{equation}

\noindent From \eqref{Equation:AlgebraicDDE}, $-Q=\dot{p}(0^{+} )-\dot{p}(0^{-} )=\dot\omega_1 (0)-\dot\omega_2 (h)$ and we have
\begin{subequations} \label{Equation:ResultingODEcCouplingConditions} 
	\begin{align} \label{Equation:ResultingODEcAlgebraic} 
		\begin{split}
			-Q=&{\omega_1}(0){{A}_{0}}+{\omega_2}(0){{A}_{1}}+{\omega_3}(0){{B}_{d}}+{\omega_4}(0){{B}_{d}} \\ 
			& +{{A}_{1}}^{\intercal}{\omega_1}(h)+{{A}_{0}}^{\intercal}{\omega_2}(h)+B_{d}^{\intercal}{\omega_5}(h)+B_{d}^{\intercal}{\omega_6}(h),
		\end{split}\\
		\mathbf{0}=&{\omega_3}(0)={\omega_4}(h)={\omega_5}(0)={\omega_6}(h), \label{Equation:ResultingODEcBoundaries3-6} \\
		\mathbf{0}=&{{\omega}_{1}}(0)-{\omega_2}(h). \label{Equation:ResultingODEcBoundaries1-2}
	\end{align} 
\end{subequations}

\noindent \Cref{Equation:ResultingODEcDynamics,Equation:ResultingODEcCouplingConditions} satisfy the ODEc \eqref{Equation:ODEc} restricted on $[0,h]$. This can be extended $\forall \tau \in \mathbb{R}$ by integrating \eqref{Equation:ResultingODEcDynamics} forward from $\tau=h$ and backward from $\tau=0$. The result is a solution to the ODEc \eqref{Equation:ODEc}.
\end{proof}

\begin{lem} \label{Lemma:ODEsolnsDDEsoln}
If $\left\{\omega _{1}^{} (\tau ),\omega _{2}^{} (\tau ),\omega _{3}^{} (\tau ),\omega _{4}^{} (\tau ),\omega _{5}^{} (\tau ),\omega _{6}^{} (\tau )\right\}$ is a solution to the ODEc \eqref{Equation:ODEc}, then there exists a solution $p(\cdot)$ to the DDEc \eqref{Equation:DDEc} on $[-h,h]$ given by
	\[\forall \tau \in [0,h], \left\{\begin{aligned} p(\tau ) & = {\textstyle\frac{1}{2}} \left[\omega _{1}^{} (\tau )+\omega _{2}^{\intercal} (h-\tau )\right] \\ 
	p(-\tau ) &= {\textstyle\frac{1}{2}} \left[\omega _{1}^{} (\tau )+\omega _{2}^{\intercal} (h-\tau )\right]^{\intercal}  \end{aligned}\right\}.\]
\end{lem}
\begin{proof}
Restrict the solution of the ODEc \eqref{Equation:ODEc} to $[0,h]$ and let $\forall \tau \in [0,h], \eta (\tau )={\textstyle\frac{1}{2}} \left[\omega _{1} (\tau )+\omega _{2}^{\intercal} (h-\tau )\right]$. Observe from \eqref{Equation:Boundary12ODE} that $\eta (0)={\textstyle\frac{1}{2}} \left[\omega _{1}^{} (0)+\omega _{1}^{\intercal} (0)\right]=\eta (0)^{\intercal} $. Construct $p(\cdot )$ on $(0,h]$ and $[-h,0)$ as follows
\begin{equation} \label{Equation:Construct-p} 
	\forall \tau \in (0,h], \left\{p(\tau )\buildrel\Delta\over= \eta (\tau ),{\rm \; \; }p(-\tau )\buildrel\Delta\over= \eta (\tau )^{\intercal} \right\},  
\end{equation} 
giving
\begin{equation} \label{Equation:Constructed-p-flip} 
	p(-\tau )=p(\tau )^{\intercal} \:\: \forall \tau \in [0,h],
\end{equation} 
\noindent which is the symmetry property \eqref{Equation:SymmetricDDE}. Note that \eqref{Equation:Constructed-p-flip} is not saying that $p(\tau )={\textstyle\frac{1}{2}} \left[\omega _{1}^{} (\tau )+\omega _{2}^{\intercal} (h-\tau )\right]$ on $[-h,h]$. Next, note that
\begin{align*}
	2\dot{p}(\tau ) =&\dot{\omega }_{1}^{{}}(\tau )-\dot{\omega }_{2}^{{}}{{(h-\tau )}^{\intercal}}, \\ 
	 =&{{\omega }_{1}}(\tau ){{A}_{0}}+{{\omega }_{2}}(\tau ){{A}_{1}}+{{\omega }_{1}}{{(h-\tau )}^{\intercal}}A_{1}^{{}}+{{\omega }_{2}}{{(h-\tau )}^{\intercal}}A_{0}^{{}} \\ 
	& +\smallint\limits_{-\tau }^{0}{{{\omega }_{1}}(\tau +\theta ){{A}_{D}}(\theta )d\theta }+\smallint\limits_{-h}^{-\tau }{{{\omega }_{2}}(\tau +\theta +h){{A}_{D}}(\theta )d\theta } \\ 
	& +\smallint\limits_{-h}^{-\tau }{{{\omega }_{1}}{{(-\tau -\theta )}^{\intercal}}{{A}_{D}}(\theta )d\theta }+\smallint\limits_{-\tau }^{0}{{{\omega }_{2}}{{(h-\tau -\theta )}^{\intercal}}{{A}_{D}}(\theta )d\theta },
\end{align*}

\noindent and hence
\begin{align*}
	\dot{p}(\tau )&=p(\tau )A_{0}^{{}}+p{{(h-\tau )}^{\intercal}}A_{1}^{{}} \\ 
	& +\smallint\limits_{-\tau }^{0}{p(\tau +\theta ){{A}_{D}}(\theta )d\theta }+\smallint\limits_{-h}^{-\tau }{p{{(-\tau -\theta )}^{\intercal}}{{A}_{D}}(\theta )d\theta }.
\end{align*}

\noindent Using \eqref{Equation:Constructed-p-flip}, we get
	\[\dot{p}(\tau )=p(\tau )A_{0}^{{}}+p(\tau -h)A_{1}^{{}}+\smallint\limits_{-h}^{0}{p(\tau +\theta ){{A}_{D}}(\theta )d\theta },\]
which is the dynamical relation \eqref{Equation:DynamicDDE}. It remains to show that
\begin{align*}
	-Q&=A_{0}^{\intercal}p(0)+p(0){{A}_{0}}+A_{1}^{\intercal}p(h)+p(-h){{A}_{1}} \\ 
	& +\smallint\limits_{-h}^{0}{\left[ {{A}_{D}}{{(\theta )}^{\intercal}}p(-\theta )+p(\theta ){{A}_{D}}(\theta ) \right]d\theta }. 
\end{align*}

\noindent This follows from the relations:
\begin{multline*}
	A_{0}^{\intercal}p(0)+p(0){{A}_{0}}+A_{1}^{\intercal}p(h)+p(-h){{A}_{1}}=\tfrac{1}{2}[\omega _{1}^{{}}(0){{A}_{0}}+\omega _{2}^{{}}(0){{A}_{1}}+ \\ 
	\text{  }A_{1}^{\intercal}\omega _{1}^{{}}(h)+A_{0}^{\intercal}\omega _{1}^{{}}(0)]+\tfrac{1}{2}[A_{0}^{\intercal}\omega _{1}^{\intercal}(0)+A_{1}^{\intercal}\omega _{2}^{\intercal}(0)\\
	+\omega _{1}^{\intercal}(h){{A}_{1}}+\omega _{1}^{\intercal}(0){{A}_{0}}],
\end{multline*}

\begin{multline*}
	\smallint\limits_{-h}^{0}{[ {{A}_{D}}{{(\theta )}^{\intercal}}p(-\theta )+p{{(-\theta )}^{\intercal}}{{A}_{D}}(\theta ) ]d\theta } = \tfrac{1}{2}\smallint\limits_{-h}^{0}{{A}_{D}}{{(\theta )}^{\intercal}}[ \omega_{1}(-\theta )	\\
	+\omega_{2}^{\intercal}(h+\theta ) ]+[ \omega_{1}^{\intercal}(-\theta )+\omega_{2}(h+\theta ) ]{{A}_{D}}(\theta )d\theta ,
\end{multline*}

\begin{align*}
	\smallint\limits_{-h}^{0}{{{A}_{D}}{{(\theta )}^{\intercal}}\omega _{1}^{{}}(-\theta )d\theta } &=B_{d}^{\intercal}\smallint\limits_{-h}^{0}{{{\left( C_{d}^{{}}{{e}^{{{A}_{d}}\theta }} \right)}^{\intercal}}\omega _{1}^{{}}(-\theta )d\theta }\\
	&=B_{d}^{\intercal}\omega _{5}^{{}}(h),
\end{align*}	
\begin{align*}
	\smallint\limits_{-h}^{0}{\omega _{2}^{{}}(h+\theta ){{A}_{D}}(\theta )d\theta } &=\smallint\limits_{-h}^{0}{\omega _{2}^{{}}(h+\theta )C_{d}^{{}}{{e}^{{{A}_{d}}\theta }}d\theta }{{B}_{d}}\\
	&={{\omega }_{4}}(0){{B}_{d}},  
\end{align*}

\noindent and using \eqref{Equation:AlgebraicODE} and \eqref{Equation:Boundary12ODE} to get
\begin{multline*}
	\tfrac{1}{2}[\omega_{1}(0){{A}_{0}}+\omega_{2}(0){{A}_{1}}+{{\omega }_{4}}(0){{B}_{d}}+{{A}_{1}}^{\intercal}\omega_{1}(h)+{{A}_{0}}^{\intercal}\omega _{1}(0) \\ 
	+B_{d}^{\intercal}\omega_{5}^{{}}(h)]+\tfrac{1}{2}[{{A}_{0}}^{\intercal}\omega _{1}^{\intercal}(0)+{{A}_{1}}^{\intercal}\omega _{2}^{\intercal}(0)+B_{d}^{\intercal}\omega _{4}^{\intercal}(0)\\
	+\omega _{1}^{\intercal}(h){{A}_{1}} \text{  }+\omega _{1}^{\intercal}(0){{A}_{0}}+\tfrac{1}{2}\omega _{5}^{\intercal}(h)B_{d}^{{}}]=-Q. 
\end{multline*}
\end{proof}

\begin{thm} \label{Theorem:DDEuniqueODEunique}
For all $Q$, the DDEc \eqref{Equation:DDEc} has a unique solution if and only if the ODEc \eqref{Equation:ODEc} has a unique solution.
\end{thm}
\begin{proof}
\textit{Sufficiency}: Assume the ODEc \eqref{Equation:ODEc} has a unique solution denoted by $\left\{\omega _{1} (\tau ),\omega _{2}(\tau ),\omega _{3}(\tau ),\omega_{4}(\tau ),\omega_{5}(\tau ),\omega _{6}(\tau )\right\}$. From \Cref{Lemma:ODEsolnsDDEsoln}, there exists a solution to the DDEc \eqref{Equation:DDEc} on $[-h,h]$ given by
	\[\forall \tau \in [0,h], \left\{ \begin{aligned}
	{{p}^{1}}(\tau )&=\tfrac{1}{2}\left[ \omega_{1}(\tau )+\omega _{2}^{\intercal}(h-\tau ) \right] \\ 
	{{p}^{1}}(-\tau )&=\tfrac{1}{2}{{\left[ \omega_{1}(\tau )+\omega _{2}^{\intercal}(h-\tau ) \right]}^{\intercal}}  
	\end{aligned} \right\}.\]

\noindent From \Cref{Lemma:FlippedUnique}, namely $\omega _{1} (\tau )=\omega _{2}^{\intercal} (h-\tau )$, we simplify to have
	\[\forall \tau \in [0,h], \left\{p^{1} (\tau )=\omega _{1} (\tau ),{\rm \; }p^{1} (-\tau )=\omega _{1}^{\intercal} (\tau )\right\}.\] 

\noindent Moreover, $p^{1} (\tau )$ is a unique solution to the DDEc \eqref{Equation:DDEc}. To see this, assume there exists an arbitrary solution to the DDEc \eqref{Equation:DDEc} denoted by $p^{2} (\tau )$. From \Cref{Lemma:DDEsolnODEsoln}, a solution to the ODEc \eqref{Equation:ODEc} on $\left[0,h\right]$ is given by
	\[\left\{ {{p}^{2}}(\tau ),{{p}^{2}}{{(h-\tau )}^{\intercal}},\ldots ,\smallint\limits_{-h+\tau }^{0}{{{\left( C_{d}^{{}}{{e}^{{{A}_{d}}\theta }} \right)}^{\intercal}}{{p}^{2}}{{(h-(\tau -\theta ))}^{\intercal}}d\theta } \right\}.\]
Uniqueness of solutions to the ODEc \eqref{Equation:ODEc} requires that on $\left[0,h\right]$, $p^{2} (\tau )=\omega _{1} (\tau )=p^{1} (\tau )$, and thus the DDEc \eqref{Equation:DDEc} has a unique solution.

\textit{Necessity}: This requires showing that if the DDEc \eqref{Equation:DDEc} has a unique solution then the ODEc \eqref{Equation:ODEc} has a unique solution. By contraposition, assume the ODEc \eqref{Equation:ODEc} has a nonunique solution. This means a nontrivial solution to the ODEc exists for $Q=\mathbf{0}$ which from \Cref{Lemma:ODEsolnsDDEsoln} can generate a solution to the DDEc \eqref{Equation:DDEc} with $Q=\mathbf{0}$ denoted by $p^{0} (\tau )$. This implies that given any arbitrary symmetric $Q$ and an associated solution $p^{1} (\tau )$, then $p^{1} (\tau )+p^{0} (\tau )$ is also a solution for the same $Q$ resulting in the DDEc \eqref{Equation:DDEc} having a nonunique solution.
\end{proof}

\begin{cor} \label{Corollary:DDEUniqueSolution}
For all $Q$, the DDEc \eqref{Equation:DDEc} has a unique solution if and only if the spectrum condition in \Cref{Definition:Spectrum} is satisfied.
\end{cor}
\begin{proof}
Follows directly from \Cref{Theorem:DDEuniqueODEunique,Theorem:UniqueODESpectrum}.
\end{proof}

\section{Analytic Solution} \label{Section:AnalyticSolution}

Since a solution to the DDEc \eqref{Equation:DDEc} can be derived from a solution to the ODEc \eqref{Equation:ODEc}, the objective of this section is to solve the ODEc analytically by writing it as an initial value problem. This requires solving a linear system of $n_{s}$ scalar equations. We assume in this section that the spectrum condition holds. Using $vec(A+B)=vec(A)+vec(B)$, $vec(ADB)=(B^{\intercal} \otimes A)vec(D)$ from \cite{Brewer1978}, \eqref{Equation:ODE} is equivalently written as
\small\begin{equation} \label{Equation:vecODEcDynamics} 
	\left[\begin{array}{c} {vec(\dot{\Omega }_{1} (\tau ))} \\ {\vdots } \\ {vec(\dot{\Omega }_{6} (\tau ))} \end{array}\right]=E\left[\begin{array}{c} {vec(\Omega _{1} (\tau ))} \\ {\vdots } \\ {vec(\Omega _{6} (\tau ))} \end{array}\right],  
\end{equation} 
and equations \eqref{Equation:AlgebraicODE}-\eqref{Equation:Boundary3456ODE} as
\small\begin{equation} \label{Equation:vecODEcCouplingConditions} 
	\left[\begin{array}{c} {-vec(Q)} \\ {0} \\ {\vdots } \\ {0} \end{array}\right]=F_{1} \left[\begin{array}{c} {vec(\Omega _{1} (0))} \\ {\vdots } \\ {vec(\Omega _{6} (0))} \end{array}\right]+F_{2} \left[\begin{array}{c} {vec(\Omega _{1} (h))} \\ {\vdots } \\ {vec(\Omega _{6} (h))} \end{array}\right],  
\end{equation} 
where $E$, $F_{1}$ and $F_{2}$ are $n_{s} \times n_{s} $ constant matrices as follows
\begin{align*}
	E&=\left[\begin{array}{ccc}
		 {A_{0}^{\intercal} \otimes I_{n} } & {A_{1}^{\intercal} \otimes I_{n} } & {B_{d}^{\intercal} \otimes I_{n} } \\
		 {-I_{n} \otimes A_{1}^{\intercal} } & {-I_{n} \otimes A_{0}^{\intercal} } & {\mathbf{0}} \\ 
		 {C_{d}^{\intercal} \otimes I_{n} } & {\mathbf{0}} & {-A_{d}^{\intercal} \otimes I_{n} }  \\ 
		 {\mathbf{0}} & {-\left(C_{d}^{} e^{-A_{d} h} \right)^{\intercal} \otimes I_{n} } & {\mathbf{0}}  \\ 
		 {I_{n} \otimes \left(C_{d}^{} e^{-A_{d} h} \right)^{\intercal} } & {\mathbf{0}} & {\mathbf{0}} \\ 
		 {\mathbf{0}} & {-I_{n} \otimes C_{d}^{\intercal} } & {\mathbf{0}}
	 \end{array}\right.\\
	 &\left.\begin{array}{ccc}
		{B_{d}^{\intercal} \otimes I_{n} } & {\mathbf{0}} & {\mathbf{0}} \\
		{\mathbf{0}} & {-I_{n} \otimes B_{d}^{\intercal} } & {-I_{n} \otimes B_{d}^{\intercal} } \\ 
		{\mathbf{0}} & {\mathbf{0}} & {\mathbf{0}} \\ 
		{-A_{d}^{\intercal} \otimes I_{n} } & {\mathbf{0}} & {\mathbf{0}} \\ 
		{\mathbf{0}} & {I_{n} \otimes A_{d}^{\intercal} } & {\mathbf{0}} \\ 
		{\mathbf{0}} & {\mathbf{0}} & {I_{n} \otimes A_{d}^{\intercal} } 
	\end{array}\right],
\end{align*}
	
	\[{{F}_{1}}=\left[ \begin{matrix}
	A_{0}^{\intercal}\otimes {{I}_{n}} & A_{1}^{\intercal}\otimes {{I}_{n}} & B_{d}^{\intercal}\otimes {{I}_{n}} & B_{d}^{\intercal}\otimes {{I}_{n}} & \mathbf{0} & \mathbf{0}  \\
	{{I}_{n}}\otimes {{I}_{n}} & \mathbf{0} & \mathbf{0} & \mathbf{0} & \mathbf{0} & \mathbf{0}  \\
	\mathbf{0} & \mathbf{0} & {{I}_{{{n}_{d}}}}\otimes {{I}_{n}} & \mathbf{0} & \mathbf{0} & \mathbf{0}  \\
	\mathbf{0} & \mathbf{0} & \mathbf{0} & \mathbf{0} & {{I}_{n}}\otimes {{I}_{{{n}_{d}}}} & \mathbf{0}  \\
	\mathbf{0} & \mathbf{0} & \mathbf{0} & \mathbf{0} & \mathbf{0} & \mathbf{0}  \\
	\mathbf{0} & \mathbf{0} & \mathbf{0} & \mathbf{0} & \mathbf{0} & \mathbf{0}  \\
	\end{matrix} \right],\]

	\[{{F}_{2}}=\left[ \begin{matrix}
	{{I}_{n}}\otimes A_{1}^{\intercal} & {{I}_{n}}\otimes A_{0}^{\intercal} & \mathbf{0} & \mathbf{0} & {{I}_{n}}\otimes B_{d}^{\intercal} & {{I}_{n}}\otimes B_{d}^{\intercal}  \\
	\mathbf{0} & -{{I}_{n}}\otimes {{I}_{n}} & \mathbf{0} & \mathbf{0} & \mathbf{0} & \mathbf{0}  \\
	\mathbf{0} & \mathbf{0} & \mathbf{0} & \mathbf{0} & \mathbf{0} & \mathbf{0}  \\
	\mathbf{0} & \mathbf{0} & \mathbf{0} & \mathbf{0} & \mathbf{0} & \mathbf{0}  \\
	\mathbf{0} & \mathbf{0} & \mathbf{0} & {{I}_{{{n}_{d}}}}\otimes {{I}_{n}} & \mathbf{0} & \mathbf{0}  \\
	\mathbf{0} & \mathbf{0} & \mathbf{0} & \mathbf{0} & \mathbf{0} & {{I}_{n}}\otimes {{I}_{{{n}_{d}}}}  \\
	\end{matrix} \right].\]

\noindent Note that using the ODE \eqref{Equation:vecODEcDynamics}, the linear system of unknowns \eqref{Equation:vecODEcCouplingConditions} becomes:
\begin{equation} \label{Equation:vecODEcCouplingConditionsCombined} 
	\left[\begin{array}{c} {-vec(Q)} \\ {0} \\ {\vdots } \\ {0} \end{array}\right]=\left(F_{1} +F_{2} e^{Eh} \right)\left[\begin{array}{c} {vec(\Omega _{1} (0))} \\ {\vdots } \\ {vec(\Omega _{6} (0))} \end{array}\right],  
\end{equation} 
\noindent where by the following \Cref{Corollary:SpectrumNonsingular}, $F_{1} +F_{2} e^{Eh} $ is nonsingular. Note that
\begin{equation} \label{Equation:vecODEcExponentialMatrix} 
	\left[\begin{array}{c} {vec(\Omega _{1} (\tau ))} \\ {\vdots } \\ {vec(\Omega _{6} (\tau ))} \end{array}\right]=e^{E\tau } \left[\begin{array}{c} {vec(\Omega _{1} (0))} \\ {\vdots } \\ {vec(\Omega _{6} (0))} \end{array}\right],  
\end{equation} 
generates $\Omega _{1} (\tau )$, hence solves $P(\tau )$ analytically.

\begin{cor} \label{Corollary:SpectrumNonsingular}
	The matrix $F_{1} +F_{2} e^{Eh}$ in \eqref{Equation:vecODEcCouplingConditionsCombined} is nonsingular if and only if the spectrum condition in \Cref{Definition:Spectrum} is satisfied.
\end{cor}
\begin{proof}
	To show necessity, assume the spectrum condition is not satisfied. Then by \Cref{Theorem:UniqueODESpectrum}, there exists a nontrivial solution for \eqref{Equation:ODEc} when $Q=\mathbf{0}$. From the Kronecker representation \eqref{Equation:vecODEcDynamics} and \eqref{Equation:vecODEcCouplingConditions} of the ODEc \eqref{Equation:ODEc}, the nontrivial solution satisfies \eqref{Equation:vecODEcCouplingConditionsCombined} for $Q=\mathbf{0}$. Therefore, $\exists v \ne \mathbf{0}: v \in \ker(F_{1} +F_{2} e^{Eh})$, and hence $F_{1} +F_{2} e^{Eh}$ is singular.
	
	Conversely, to show sufficiency, assume $\exists v \ne \mathbf{0}: v \in \ker(F_{1} +F_{2} e^{Eh})$. Then from \eqref{Equation:vecODEcDynamics} and \eqref{Equation:vecODEcCouplingConditions}, there exists a nontrivial solution for \eqref{Equation:ODEc} when $Q=\mathbf{0}$, and by \Cref{Theorem:UniqueODESpectrum}, the spectrum condition is not satisfied.
\end{proof}

An example incorrectly handled in \cite{SCL2006}, as noted in Section \ref{Section:OriginsProblem}, is treated next.

\textbf{Example 1}: Consider the following linear time-delay system
\begin{align*}
	\dot{x}(t)&={{A}_{0}}x(t)+{{A}_{1}}x(t-1)+\smallint\limits_{-1}^{0}{{{A}_{D}}(\theta )x(t+\theta )d\theta }, \\ 
	{{A}_{0}}&=\left[ \begin{matrix}
	-1 & 0  \\
	0 & -1  \\
	\end{matrix} \right], {{A}_{1}}=\left[ \begin{matrix}
	0 & 1  \\
	-1 & 0  \\
	\end{matrix} \right], \\
	{{B}_{0}}&=\left[ \begin{matrix}
	0.3 & 0  \\
	0 & 0.3  \\
	\end{matrix} \right],\text{ }{{B}_{1}}=\left[ \begin{matrix}
	0 & 1  \\
	-1 & 0  \\
	\end{matrix} \right]{{B}_{0}}, \\ 
	 {{A}_{D}}(\theta ) &=\sin (\pi \theta ){{B}_{0}}+\cos (\pi \theta ){{B}_{1}}.
\end{align*}

\noindent The system is stable and the spectrum condition holds, which can be shown by plotting the spectrum of the time-delay system. To use the ODEc \eqref{Equation:ODEc}, put $A_{D} (\theta )$ in the form $A_{D} (\theta )=C_{d}^{} e^{A_{d} \theta } B_{d}^{} $ as follows
	\[{{A}_{d}}=\pi \left[ \begin{matrix}
	0 & -1  \\
	1 & 0  \\
	\end{matrix} \right],\text{   }{{e}^{{{A}_{d}}\theta }}=\left[ \begin{matrix}
	\cos (\pi \theta ) & -\sin (\pi \theta )  \\
	\sin (\pi \theta ) & \cos (\pi \theta )  \\
	\end{matrix} \right],\text{  }C_{d}^{{}}={{I}_{2}}\text{,  }B_{d}^{{}}={{B}_{1}}.\]

\noindent From \eqref{Equation:vecODEcCouplingConditionsCombined}, and letting $Q=I_{2} $, we get $vec(\Omega _{3} (0))=vec(\Omega _{5} (0))=0_{4\times 1} $ and
\begin{align*}
	vec({{\Omega }_{1}}(0)) &={{\left[ \begin{matrix} 0.7072 & 0 & 0 & 0.7072  \end{matrix} \right]}^{\intercal}}, \\ 
	vec({{\Omega }_{2}}(0)) &={{\left[ \begin{matrix} 0.2636 & 0.3165 & -0.3165 & 0.2636  \end{matrix} \right]}^{\intercal}}, \\ 
	\text{ }vec({{\Omega }_{4}}(0)) &={{\left[ \begin{matrix}	0.1909 & -0.3642 & 0.3642 & 0.1909 	\end{matrix} \right]}^{\intercal}}, \\ 
	vec({{\Omega }_{6}}(0)) &={{\left[ \begin{matrix} -0.1909 & 0.3642 & -0.3642 & -0.1909 	\end{matrix} \right]}^{\intercal}}.  
\end{align*}

\noindent Solving the ODEc via \eqref{Equation:vecODEcExponentialMatrix} we get $P(\tau )$ on $[0,h]$ as shown in Figure 1. 
\begin{figure}[h]
	\centering
	\includegraphics*[width=3.54in, height=1.80in, keepaspectratio=false, trim=0.00in 0.00in 0.26in 0.00in]{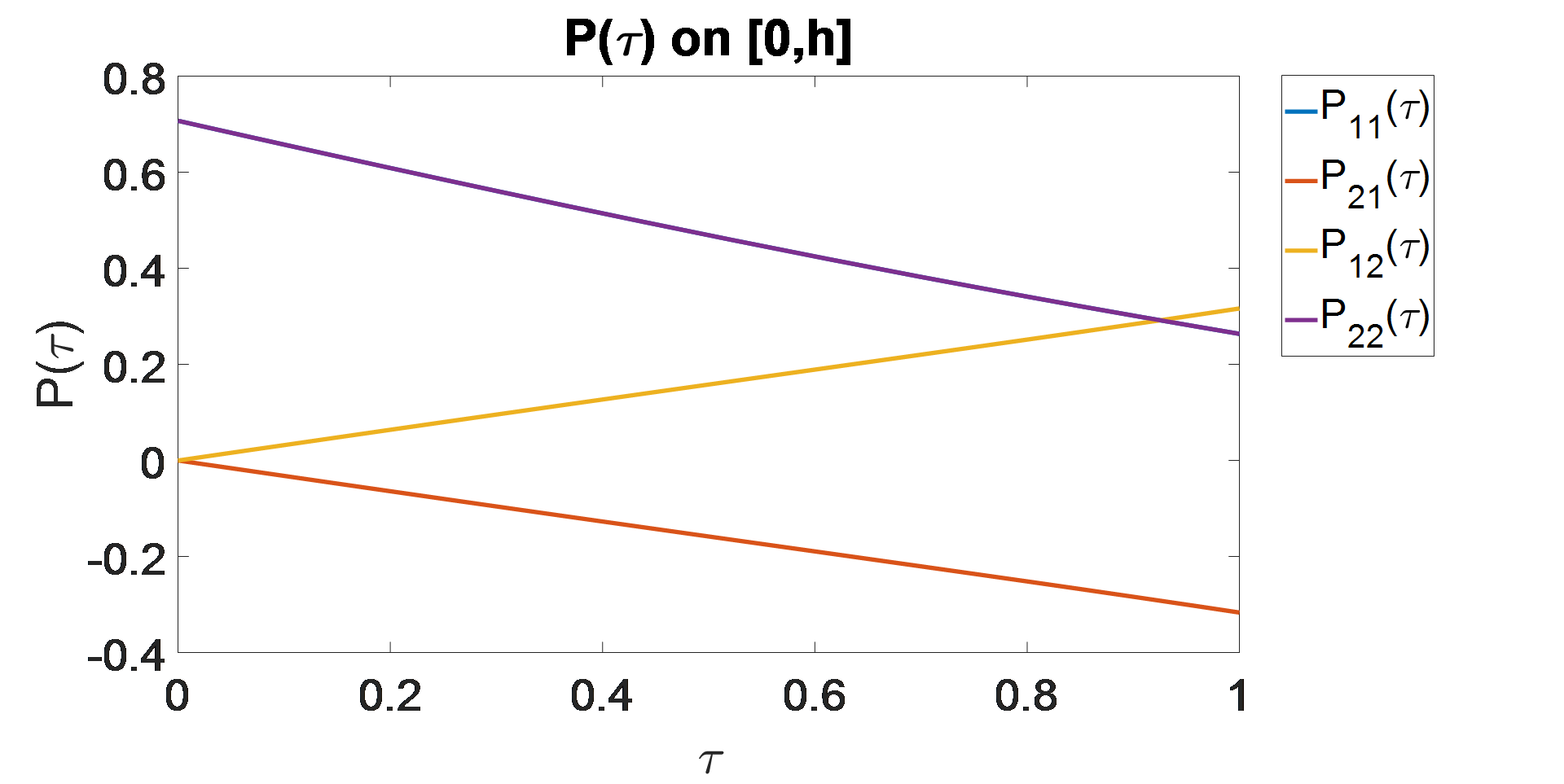}
	\caption{$P(\tau )$ as determined from $\Omega _{1} (\tau )$ in \eqref{Equation:vecODEcExponentialMatrix}.}
\end{figure}

\noindent Note that $P(0)={{\Omega }_{1}}(0)=0.7072\times I_{2}$. From \eqref{Equation:Cost-to-go_P}, $V(\phi )=\phi (0)^{\intercal} P(0)\phi (0)$ is the cost over the trajectories of \eqref{Equation:TimeDelaySystem} assuming an initial function such that $\phi (0)\ne 0$ and $\phi (\theta )=0$ for $h\le \theta <0$. This is validated by writing $y(t)=\int_{-1}^{0}e^{A_{d} \theta } B_{d}^{} x(t+\theta )d\theta  $ to get
\begin{equation} \label{Equation:MATLABValidated}
	\begin{aligned}
		 \dot{x}(t)&={{A}_{0}}x(t)+C_{d}y(t)+{{A}_{1}}x(t-1), \\ 
		 \dot{y}(t)&={{B}_{d}}x(t)-{{A}_{d}}y(t)-{{e}^{-{{A}_{d}}}}{{B}_{d}}x(t-1).
	\end{aligned}
\end{equation} 
One can then use a delay differential equation solver, such as \verb|dde23| of MATLAB in our case, to compute the cost-to-go of \eqref{Equation:MATLABValidated} for different $\phi (0)$ to validate $V(\phi )$, and thus implicitly $P(0)$.

\begin{rem} \label{Remark:ConstantDelayLike}
\Cref{Equation:TimeDelaySystem} can be written as \eqref{Equation:MATLABValidated} for $A_{D} (\theta )=C_{d}^{} e^{A_{d} \theta } B_{d}^{} $. One may use existing results for a constant delay system, see \cite{KharitonovPlischke2006}, to obtain an analytic solution $\hat{P}(\tau )$ for the resulting DDEc system whose DDE is double the rows and double the columns of the original DDE \eqref{Equation:DynamicDDE}. However, writing $P(\tau )$ of the DDE \eqref{Equation:DynamicDDE} in terms of $\hat{P}(\tau )$ is not obvious.
\end{rem}

\begin{rem} \label{Remark:AliseykoResult}
In recent work \cite{AliseykoExponentialKernel2017}, see also \cite{AliseykoPiecewiseConstant2017}; the author comments on the insufficiency of the boundary conditions provided in \cite{SCL2006}. To address existence and uniqueness, the author adds a new group of boundary conditions in the form of integral constraints. The result therefore is no longer an ordinary differential equation system with algebraically related split-boundary conditions.
\end{rem}

\section{Conclusion} \label{Section:Conclusion}
This work provides a method to compute the quadratic cost functional \eqref{Equation:Cost-to-go_P} for systems with distributed delays. Future directions include extensions to multiple commensurate delays and generalized quadratic cost functionals suitable for optimal control applications.


%

%

\section*{Acknowledgment}
The authors wish to acknowledge the communication with Vladimir L. Kharitonov pertaining to \cite{SCL2006}. We would like to also thank the anonymous reviewers, and the Associate Editor, for their valuable commentary and technical feedback. The second author would like to acknowledge the support provided by the Toyota Research Institute through the Toyota-CSAIL Joint Research Center.

\ifCLASSOPTIONcaptionsoff
  \newpage
\fi

\end{document}